\documentclass[preprint,12pt,3p]{elsarticle}

\usepackage{amssymb}
\usepackage{amsthm}

\usepackage[utf8]{inputenc}
\usepackage[english]{babel}
\usepackage{graphicx}

\usepackage{array}
\usepackage{longtable}
\usepackage{booktabs}
\usepackage{caption}
\usepackage{subcaption}
\usepackage{multirow}
\usepackage{mathtools}
\usepackage{color}

\usepackage{blindtext}
\usepackage{amssymb}

\usepackage{amsthm}
\usepackage{amsmath}
\usepackage{verbatim}
\usepackage{cleveref}

\newtheorem{theorem}{Theorem}[section]

\newcounter{example}[section]

\theoremstyle{remark}

\theoremstyle{definition}

\biboptions{sort&compress}
\usepackage{lineno}
\begin{document}

\begin{frontmatter}

\title{Periodic culling outperforms isolation and vaccination strategies in controlling Influenza A (H5N6) outbreaks in the Philippines} % \texttt{elsarticle} class\tnoteref{label0}
%\tnotetext[label0]{This is only an example}

\author[label1,label2]{Abel Lucido}
\author[label3]{Robert Smith?}
\author[label1]{Angelyn Lao}
\address[label1]{Mathematics \& Statistics Department, De La Salle University, 2401 Taft Avenue, 0922 Manila, Philippines}
\address[label2]{Department of Science \& Technology - Science Education Institute, Bicutan, Taguig, Philippines}
\address[label3]{Department of Mathematics, University of Ottawa, 585 King Edward Ave Ottawa, ON K1S 0S1, Canada}
%\address[label2]{Mathematical \& Statistical modelling Research Unit, Center for Natural Sciences and Environmental Research, 2401 Taft Avenue, 0922 Manila, Philippines\fnref{label4}}
\ead{angelyn.lao@dlsu.edu.ph}

%\cortext[cor1]{I am corresponding author}
%\author[label1,label2]{Angelyn Lao\corref{cor1}}
%\ead{angelyn.lao@dlsu.edu.ph}

\begin{abstract}
Highly Pathogenic Avian Influenza A (H5N6) is a mutated virus of Influenza A (H5N1) and a new emerging infection that recently caused an outbreak in the Philippines. The 2017 H5N6 outbreak resulted in a depopulation of 667,184 domestic birds. In this study, we incorporate half-saturated incidence in our mathematical models and investigate three intervention strategies against H5N6:  isolation with treatment, vaccination and modified culling.  We determine the direction of the bifurcation when $\mathcal{R}_0 = 1$ and show that all the  models exhibit forward bifurcation. We administer optimal control and perform numerical simulations to compare the consequences and implementation cost of utilizing different intervention strategies in the poultry population. Despite the challenges of applying each control strategy, we show that culling both infected and susceptible birds is a better control strategy in prohibiting an outbreak and avoiding further recurrence of the infection from the population compared to confinement and vaccination.
\end{abstract}

\begin{keyword}
Influenza A (H5N6) \sep half-saturated incidence \sep isolation \sep culling \sep vaccination \sep bifurcation \sep optimal control
\end{keyword}

\end{frontmatter}

% \linenumbers

\section{Introduction}
\label{sec1}

Avian influenza is a highly contagious disease of birds caused by infection with influenza A viruses that circulate in domestic and wild birds \cite{who_influenza_2018}. Some avian influenza virus subtypes are H5N1, H7N9 and H5N6, which are classified according to combinations of different virus surface proteins hemagglutinin (HA) and neuraminidase (NA). This disease is categorized as either Highly Pathogenic Avian Influenza (HPAI), which causes severe disease in poultry and results in high death rates, or Low Pathogenic Avian Influenza (LPAI), which causes mild disease in poultry \cite{who_influenza_2018}.

As reported by the World Health Organization (WHO) \cite{who_influenza_2018}, H5N1 has been detected in poultry, wild birds and other animals in over 30 countries and has caused 860 human cases in 16 of these countries and 454 deaths. H5N6 was reported emerging from China in early May 2014 \cite{joob_h5n6_2015}. H5N6 is a mutated virus of H5N1, which has been spreading in Southeast Asia since 2003 \cite{noauthor_analysis:_2016}. Bi {\em et al.}\ reported that H5N6 has replaced H5N1 as one of the dominant avian influenza virus subtypes in southern China \cite{bi_genesis_2016}. In August 2017, cases of H5N6 in the Philippines resulted in the culling of 667,184 chicken, ducks and quails \cite{noauthor_culling_2017,noauthor_president_2017}.

Due to possible threat of avian influenza virus to cause a pandemic, several mathematical models have been developed in order to test control strategies. Several included saturation incidence, where the rate of infection will eventually saturate, showing that protective measures have been put into place as the number of infected birds increases \cite{capasso_generalization_1978,liu_global_2015}.  With half-saturated incidence, it includes the half-saturation constant which pertains to the density of infected individuals that yields 50\% chance of contracting the disease\cite{shi_dynamics_2019}. Some intervention strategies employed to protect against avian influenza are biosecurity, quarantine, control in live markets, vaccination and culling. Culling is a widely used control strategy during an outbreak of avian influenza. Gulbudak {\em et al.}\ utilized a function to represent the culling rate considering both HPAI and LPAI  \cite{gulbudak2014coexistence,gulbudak_forward_2013}. The two-host model of Liu and Fang (2015) showed that screening and culling of infected poultry is a critical measure for preventing human A(H7N9) infections in the long term \cite{liu_modeling_2015}.

Emergency vaccination, prophylactic or preventive vaccination, and routine vaccination are the three vaccination strategies mentioned by the United Nations Food and Agriculture Organization (UNFAO) \cite{fao_global_2007}. In China, A(H5N1) influenza infection caused severe economic damage for the poultry industry, and vaccination served a significant role in controlling the spread of this infection since 2004 \cite{chen_avian_2009}. UNFAO and Office International des Epizooties (OIE) of the World Organization for Animal Health suggested vaccination of flocks should replace mass culling of poultry as primary control strategy during outbreak \cite{butler_vaccination_2005}. For this reason, many mathematical models focus on how vaccination could prohibit the spread of infection. 

The importance of optimal control in modelling infectious diseases has been highlighted by several recent studies. Agusto used optimal control and cost-effective analysis in a two-strain avian influenza model \cite{agusto_optimal_2013}. Jung {\em et al.}\ used optimal control in modelling H5N1 in figuring out the prevention of influenza pandemic \cite{jung_optimal_2009}. Kim {\em et al.}\ utilized an optimal-control approach in modelling tuberculosis (TB) in the Philippines \cite{kim_mathematical_2018}. Okosun and Smith?\ used optimal control to examine strategies for malaria--schistosomiasis coinfection \cite{OkosunSmith?}.

 \section{The models}
We examine three control strategies: isolation, culling and vaccination. Our mathematical models are in the form of half-saturated incidence (HSI), we take into consideration the density of infected individuals in the population that yields 50\% chance of contracting avian influenza. We present four mathematical models: a model without control, which describes the transmission dynamics of avian influenza in bird population (i.e., the avian influenza virus (AIV) model), and three models obtained from the AIV model by applying the intervention strategies isolation, vaccination and culling. Mathematical models with half-saturated incidence are more realistic compared to models with bilinear incidence  \cite{liu_global_2015, chong_mathematical_2014, lee_transmission_2018}. Description of variables and parameters used in the models are listed in the table in \ref{tab:list}.

\subsection{AIV model without intervention strategy}

\begin{figure}[ht]
\centerline{\includegraphics[width=3.5in]{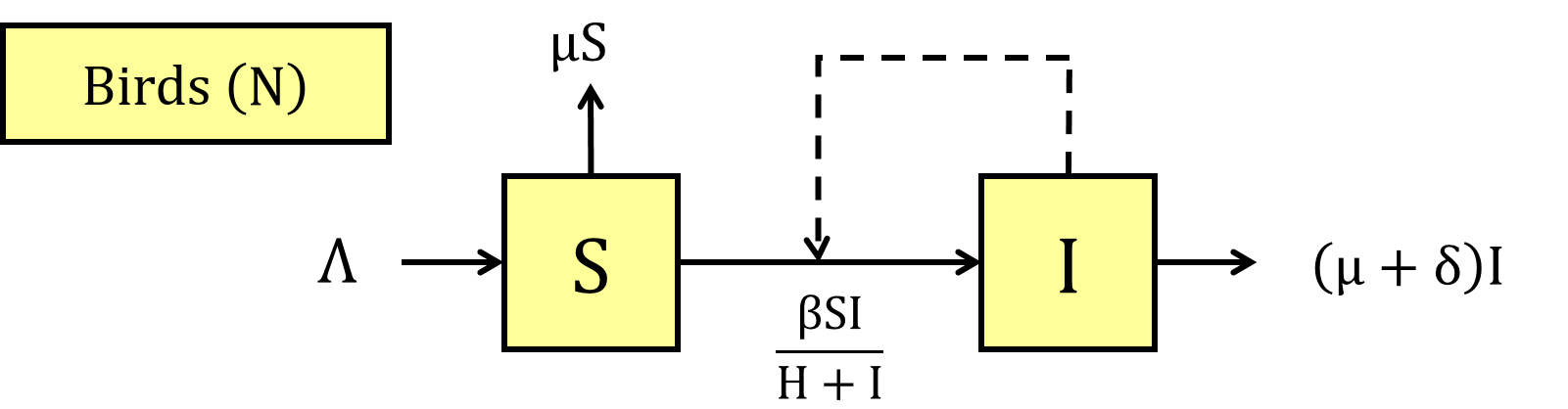}}
\vspace*{8pt}
\caption{Schematic diagram of the AIV model with half-saturated incidence.}
  \label{fig:avian model}
\end{figure}

In the AIV model without intervention strategy (shown in Fig.~\ref{fig:avian model}), the bird population is divided into sub-populations (represented by compartments): the susceptible birds ($S$) and the infected birds ($I$). The total population of birds are represented by $N(t)$ at time $t$,  where $N(t)=S(t) + I(t)$. The number of susceptible birds increases through birth rate ($\Lambda$) and reduces through the natural death rate of birds ($\mu$). Infected birds additionally decrease through the disease-specific death rate caused by the virus ($\delta$). 

The number of susceptible birds who become infected through direct contact is represented by $\frac{\beta S I}{H + I}$, which denotes the transfer of the susceptible bird population to the infected bird population. Note that $\beta$ is the rate at which birds contract avian influenza and $H$ is the half-saturation constant, indicating the density of infected individuals in the population that yields 50\% possibility of contracting avian influenza \cite{chong_mathematical_2014}. The saturation effect of the infected bird population indicates that a very large number of infected may tend to reduce the number of contacts per unit of time  due to awareness of farmers to the disease \cite{capasso_generalization_1978}. In Figure \ref{fig:avian model}, the dashed directional arrow from $I$ to the arrow from $S$ to $I$ indicates that $\frac{\beta S I}{H + I}$ is regulated by $I$.

Based on AIV model described above, we have the following system of nonlinear ordinary differential equations (ODEs):
\begin{equation}
\label{eq:avian model}
\begin{split}
\dot{S} &= \Lambda - \mu S - \dfrac{\beta S I}{H + I},\\
\dot{I} &=\dfrac{\beta S I}{H + I} - (\mu + \delta)I.
\end{split}
\end{equation}

\subsection{Confinement strategy for infected poultry (isolation model)}

\begin{figure}[ht]
\centerline{\includegraphics[width=3.75in]{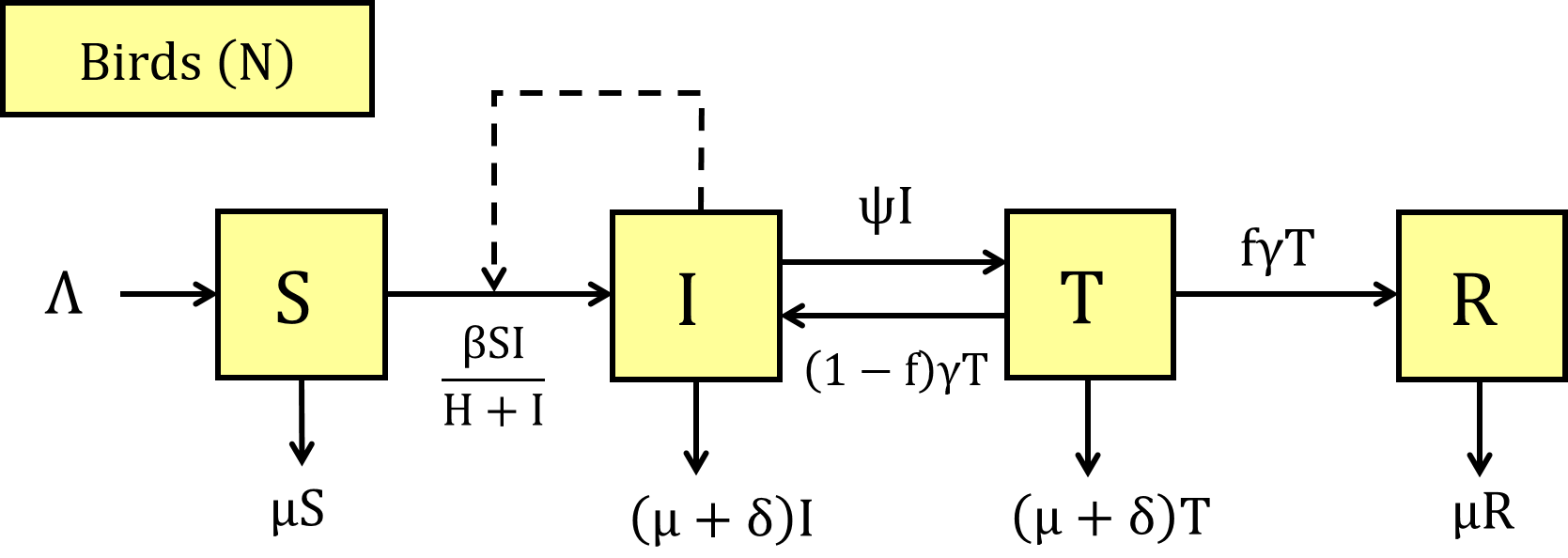}}
\vspace*{8pt}
\caption{Schematic diagram of confinement or isolation model with HSI}
  \label{fig:isolation model}
\end{figure}

Here, we employ the strategy of confining the infected poultry population (which will be referred as the isolation strategy) into the AIV model. Several studies concluded that reducing the contact rate is an effective measure in preventing the spread of infection into the population \cite{teng_contact_2018,lee_transmission_2018}. For the isolation model (shown in Fig.~\ref{fig:isolation model}), we have included the compartment representing the population of isolated birds that undergoes treatment ($T$) and the compartment representing the population of recovered birds ($R$). We denote the isolation rate of identified infected birds by $\psi$ and the release of birds from isolation by $\gamma$.

We apply treatment to birds during isolation so that some birds can be released from isolation even though they were still infected. The proportion of infected birds that have been put into isolation and recovered is represented by $f$; infected birds that have not recovered and remained infected are represented by $(1-f)$. We did not consider natural recovery of poultry in our model due to high mortality rate of HPAI virus infection.

The system of ODEs for the isolation model is 
\begin{equation}
\label{eq:iso model}
\begin{split}
        \dot{S} &= \Lambda - \mu S - \dfrac{\beta S I}{H + I} , \\
        \dot{I} &=\dfrac{\beta S I}{H + I} + (1-f)\gamma T - (\mu + \delta + \psi)I , \\
        \dot{T} &= \psi I - (\mu + \delta + \gamma)T, \\ 
        \dot{R} &=f\gamma T - \mu R.
\end{split}
\end{equation}

\subsection{Immunization of the poultry population (vaccination model)}
\begin{figure}[ht]
\centerline{\includegraphics[width=3.25in]{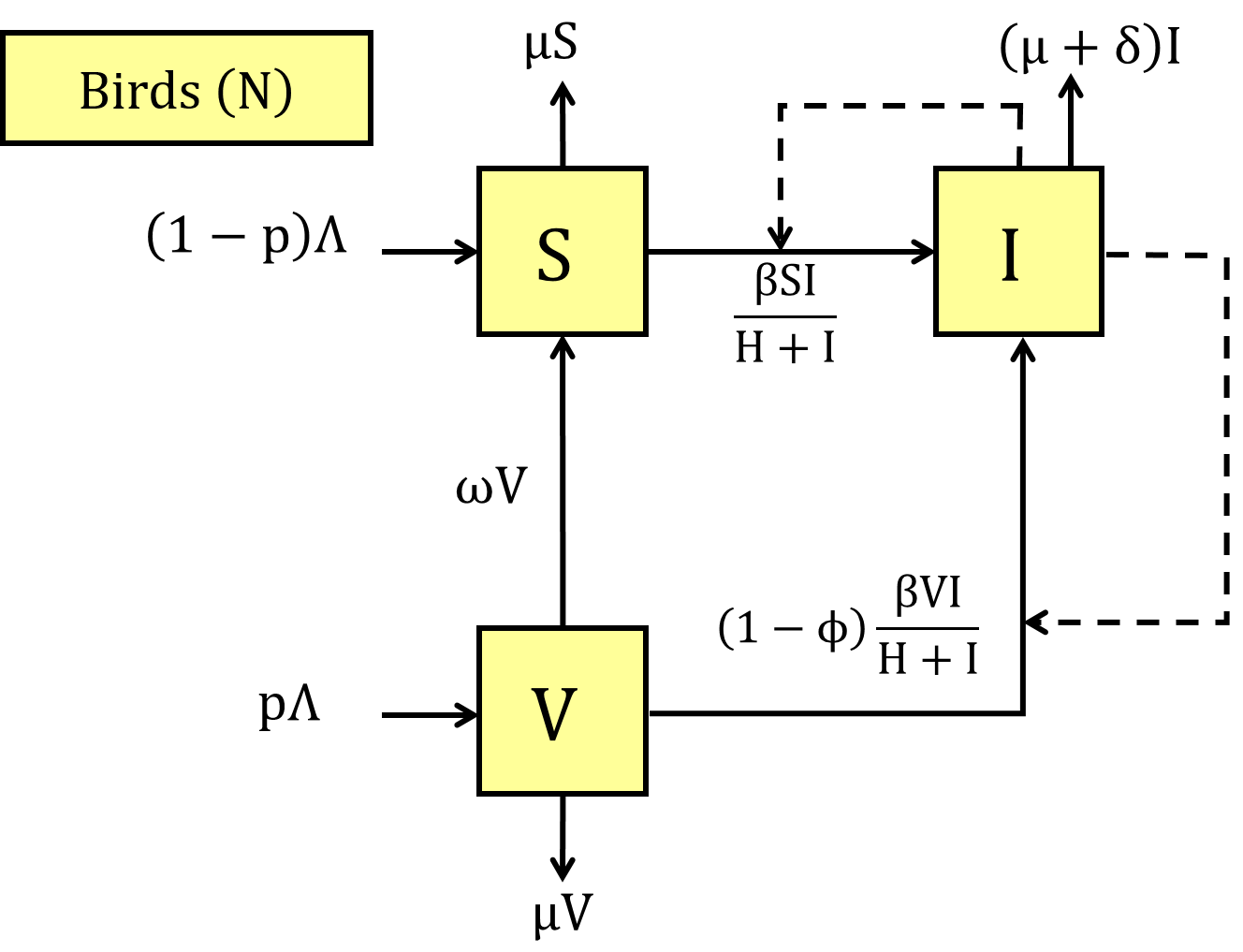}}
\vspace*{8pt}
\caption{Schematic diagram of preventive vaccination model with HSI}
  \label{fig:vaccination model}
\end{figure}

Infection-reduction measures can also help in controlling the disease during an outbreak~\cite{gumel_global_2009}. Aside from isolating infected birds, we also recognize vaccination of susceptible birds as a control strategy to reduce the number of infected birds.  Joob and Viroj \cite{joob_h5n6_2015} reported the existence and effectiveness of a vaccine for birds with H5N6. 

According to UNFAO, prophylactic vaccination (or preventive vaccination) is carried out if a high risk of virus incursion is identified and early detection or rapid response measures may not be sufficient \cite{fao_global_2007}. In this view, we modified the vaccination model (presented in \cite{lee_transmission_2018}) by dividing the birth rate ($\Lambda$) depending on the prevalence rate of vaccination ($p$) as shown in Fig.~\ref{fig:vaccination model} \cite{lee_transmission_2018}. The poultry population prone to H5N6 is divided into two compartments: the vaccinated birds represented by $V$ and the susceptible or unvaccinated birds denoted by $S$. In our vaccination model, we differentiate the immunized group (vaccinated) from non-immunized  group (unvaccinated). 

We investigate the effectiveness of the vaccine not only through its reported efficacy (denoted by $\phi$) but also based on the waning rate of the vaccine (denoted by $\omega$). To represent the acquired immunity of the vaccinated group, the infectivity of vaccinated birds is reduced by a factor $1-\phi$. The system of ODEs representing the vaccination model is
\begin{equation}
\label{eq:vac model}
\begin{split}
   \dot{S} &= (1-p) \Lambda + \omega V - \mu S - \dfrac{\beta S I}{H + I}, \\ 
\dot{V} &= p \Lambda - (\mu + \omega) V - (1-\phi) \dfrac{\beta V I}{H + I}, \\ 
\dot{I} &=\dfrac{\beta S I}{H + I} + (1-\phi) \dfrac{\beta V I}{H + I} - (\mu + \delta)I.
\end{split}
\end{equation}

\subsection{Depopulation of susceptible and infected birds (culling model)}

\begin{figure}[ht]
\centerline{\includegraphics[width=3.25in]{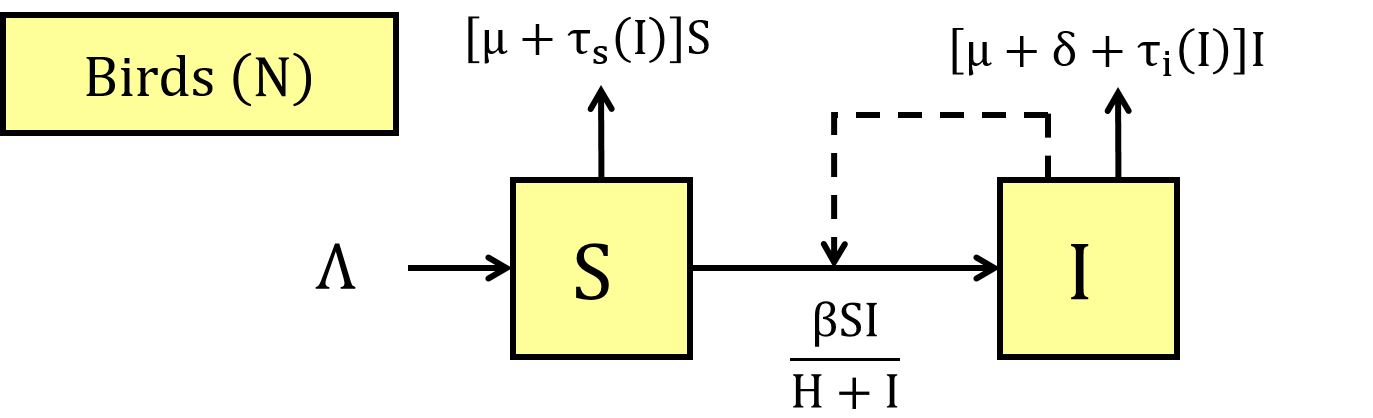}}
\vspace*{8pt}
\caption{Schematic diagram of depopulation or culling model with HSI}
  \label{fig:culling model}
\end{figure}

During outbreaks of avian influenza, one of the most widely used strategies is depopulation or culling \cite{fao_global_2007}. A total number of $667,184$ chicken, ducks, and quails were culled in August 2017 to resolve the outbreak of H5N6 in the Philippines \cite{noauthor_culling_2017,noauthor_president_2017}. Several studies employed culling as a control strategy against avian influenza \cite{chong_avian-only_2016, chong_modeling_2015, liu_modeling_2015, gulbudak_forward_2013} and pointed out the significance of obtaining an appropriate threshold policy to combat avian influenza and prevent the overkilling of birds.

The culling models of Gulbudak {\em et al.}\ considered bilinear incidence in transmission of infection \cite{gulbudak_forward_2013,gulbudak2014coexistence}. Gulbudak and Martcheva designated a culling rate for each strategy by a different function \cite{gulbudak_forward_2013}. Gulbudak {\em et al.}\ used half-saturated incidence to represent the culling rate for the infected population \cite{gulbudak2014coexistence}. In our case, we improved their culling model by incorporating the dynamics of half-saturated incidence on the transmission of infection and on the culling rate for infected birds and for susceptible birds that are at high risk of infection.

Moreover, we define the culling function of the infected and susceptible birds as $\tau_i (I) = \frac{c_i I}{H + I}$ and $\tau_s (I) = \frac{c_s I}{H + I}$, respectively. The culling rate is represented by $c_s$ for susceptible birds and $c_i$ for infected birds. The following system of ODEs represents the culling model:
\begin{equation}
\label{eq:cul model}
\begin{split}
        \dot{S} &= \Lambda - \mu S - \tau_s (I)S - \dfrac{\beta S I}{H + I}, \\ 
        \dot{I} &=\dfrac{\beta S I}{H + I} - (\mu + \delta)I - \tau_i (I) I.
\end{split}
\end{equation}

\section{Stability and bifurcation analysis}

We first analyze the AIV model without intervention. The disease-free equilibrium (DFE) of the AIV model (\ref{eq:avian model}) is
 \begin{equation*}
 \label{eq:avianDFE}
    E_A^0 = \left(S^0, I^0 \right) = \left(\dfrac{\Lambda }{\mu}, 0\right).
 \end{equation*}

The basic reproduction number for the AIV model is
\begin{align}
\label{eq:avian RA}
   \mathcal{R}_A= \dfrac{\beta \Lambda}{H \mu (\mu + \delta)}.
\end{align}
The disease-free equilibrium $E_A^0$ of the AIV model is locally asymptotically stable if $\mathcal{R}_A < 1$ and unstable if $\mathcal{R}_A > 1$.

The endemic equilibrium for the AIV model is represented by
\begin{align}
\label{eq:avian endemic}
E_A^* = \left(S^*, I^* \right) = \left(\dfrac{\Lambda + H (\mu + \delta)}{\mu + \beta}, \dfrac{\beta \Lambda - \mu H (\mu + \delta)}{(\mu + \delta) (\mu + \beta)} \right).    
\end{align} From AIV model, we obtain two possible endemic equilibria, that is $E_A^*$ and $$E_{A_1}^*~=~(S_1^*, I_1^*)~=~ \left( \dfrac{(\mu + \delta)(H+I_1^*)}{\mu + \beta}, \dfrac{\Lambda-\mu S_1^*}{\beta -(\Lambda-\mu S^*_1)} \right).$$ Simplifying $S_1^*$ and $I_1^*$ will result to $E_{A_1}^* = E_A^*$, and we have an endemic equilibrium. We can rewrite $I^*$ as
\begin{equation*}
    I^*  =  \dfrac{\mu H}{\mu + \beta} (\mathcal{R}_A - 1).
\end{equation*}
Hence when $\mathcal{R}_A \leq 1$ then $I^* \leq 0$, so there is no biologically feasible endemic equilibrium. For $\mathcal{R}_A > 1$, we have $I^* > 0$, so we have an endemic equilibrium. We conclude that the AIV model has no endemic equilibrium when $\mathcal{R}_A \leq 1$, and has an endemic equilibrium when $\mathcal{R}_A > 1$. It follows that reducing the basic reproduction number $(\mathcal{R}_A)$ below one is sufficient to eliminate avian influenza from the poultry population.

\begin{figure}[ht]
 \centering
  \includegraphics[width=3in]{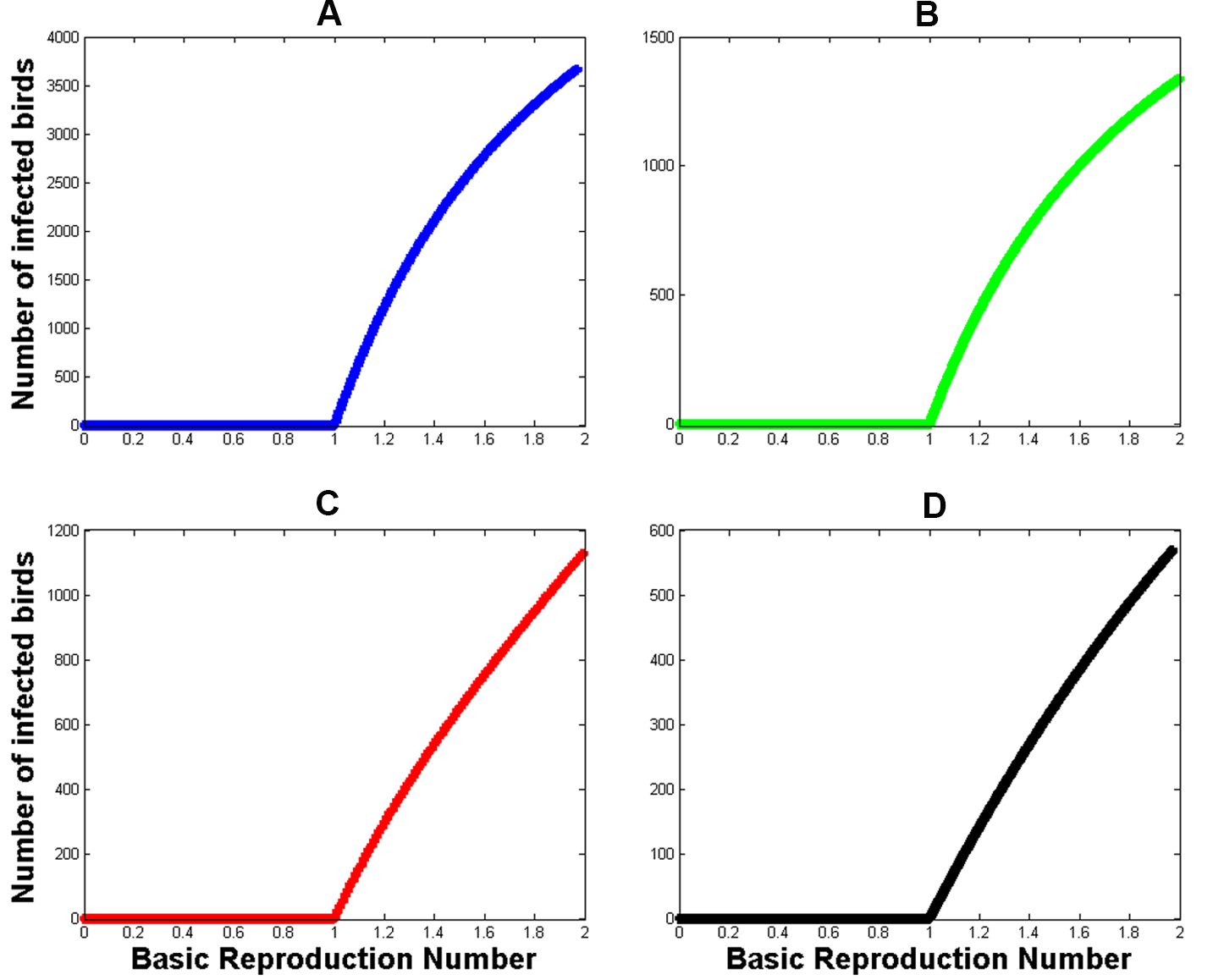}
  \caption{Bifurcation diagram for AIV (A), isolation (B), vaccination (C) and culling (D) model with respect to their basic reproduction number, indicating only forward bifurcations.}
  \label{fig:bifurcation}
\end{figure}

As exhibited in Fig.~\ref{fig:bifurcation}{A}, we have a bifurcation plot between the infected population and the basic reproduction number $\mathcal{R}_A$. Clearly, we have a forward bifurcation for the AIV model, showing that when the basic reproduction number crosses unity, an endemic equilibrium appears. 

We continue by investigating different strategies that can reduce or stop the spreading of AIV. From the isolation model (\ref{eq:iso model}), the DFE is given by 
\begin{equation*}
\label{eq:iso DFE}
    E_T^0 = \left(S^0, I^0, T^0, R^0 \right) = \left(\dfrac{\Lambda }{\mu}, 0, 0 , 0 \right).
\end{equation*}

The corresponding basic reproduction number ($\mathcal{R}_T$) is represented by
\begin{align}
\label{eq:iso RT}
   \mathcal{R}_T  =  \dfrac{\beta \Lambda (\mu + \delta + \gamma)}{H \mu [(\mu + \delta + \psi) (\mu + \delta + \gamma) - (1-f) \gamma \psi]}.
\end{align}

The disease-free equilibrium ($E_T^0$) of the isolation model is locally asymptotically stable if $\mathcal{R}_T < 1$ and unstable if $\mathcal{R}_T > 1$.  Consequently, we can identify some conditions on how confinement of infected birds affects the stability of $E_T^0$. The DFE ($E_T^0$) is locally asymptotically stable whenever 
$$\psi > \dfrac{\beta \Lambda (\mu + \delta + \gamma) - H \mu (\mu + \delta)(\mu + \delta + \gamma)}{H \mu (\mu + \delta + f \gamma)}. $$

For the endemic equilibrium of the isolation model (\ref{eq:iso model}), we indicate the presence of infection in the population by letting $I \neq 0$ and solve for $S, I, T$, and $R$. Thus, we have
\begin{equation}
\label{eq:iso endemic}
\begin{split}
E_T^*  & =  \left(S^{**}, I^{**}, T^{**}, R^{**} \right) \\ 
& =  \left(\dfrac{\Lambda (H + I^{**})}{\mu (H + I^{**}) + \beta I^{**}},  \dfrac{[\beta \Lambda - \mu H (\mu + \delta +\psi)](\mu + \delta + \gamma) + (1-f) \gamma \psi \mu H}{(\mu + \beta)[(\mu + \delta + \psi)(\mu + \delta + \gamma)-(1-f) \gamma \psi ]}, \right. \\ 
& \quad \left.\dfrac{\psi I^{**}}{\mu + \delta + \gamma},  \dfrac{f \gamma \psi I^{**}}{\mu (\mu + \delta + \gamma)} \right).
\end{split}
\end{equation}

Given the basic reproduction number (\ref{eq:iso RT}), we rewrite the expression $I^{**}$ of the isolation model as 

\begin{equation}
\label{eq:iso Ib solution}
\begin{aligned}
 I^{**} & = \dfrac{\mu H (\mathcal{R}_T - 1)}{\mu + \beta}.
\end{aligned}
\end{equation}

From (\ref{eq:iso Ib solution}), it follows that when $\mathcal{R}_T \leq 1$, we have $I_b \leq 0$ and there is no endemic equilibrium; but when $\mathcal{R}_T > 1$, we get $I_b>0$ and we have an endemic equilibrium. Thus, the isolation model (\ref{eq:iso model}) has no endemic equilibrium when $\mathcal{R}_T \leq 1$ and has an endemic equilibrium when $\mathcal{R}_T > 1$. Hence there is no backward bifurcation for the isolation model when $\mathcal{R}_T<1$.

In Fig.~\ref{fig:bifurcation}{B}, we have a forward bifurcation for the isolation model, which supports our claim. The bifurcation plot between the infected population $I^{**}$ and the basic reproduction number $\mathcal{R}_T$ for the isolation model shows that reducing $\mathcal{R}_T$ below unity is enough to eliminate avian influenza from the poultry population.

Next we analyze the stability of the associated equilibria of the AIV model with vaccination strategy (\ref{eq:vac model}). The DFE and the basic reproduction number are 
\begin{equation*}
\label{eq:vac DFE}
    E_V^0 = (S^0, V^0, I^0) = \left(\frac{(\mu + \omega - p \mu) \Lambda}{\mu (\mu + \omega)} , \frac{p \Lambda}{\mu + \omega}, 0 \right)
\end{equation*}
and
\begin{equation*}
\label{eq:vacRV}
    \mathcal{R}_V =  \frac{\Lambda \beta(\mu + \omega - p \mu \phi)}{\mu H(\mu + \omega) (\mu + \delta)}.
\end{equation*}

The disease-free equilibrium $E_V^0$ of vaccination model is locally asymptotically stable if $\mathcal{R}_V < 1$ and unstable if $\mathcal{R}_V > 1$. Moreover, we obtain some conditions for the prevalence rate of vaccination ($p$) and vaccine efficacy ($\phi$), which both range from $0$ to $1$. The DFE $E_V^0$ of vaccination model is locally asymptotically stable whenever $$\dfrac{(\mu + \omega)}{\mu} \left( 1 - \dfrac{\mu H (\mu + \delta)}{\Lambda \beta} \right) \leq p \phi < 1$$

For the endemic equilibrium of the vaccination model (\ref{eq:vac model}), we obtain the following:
\begin{equation*}
\label{eq:vac endemic}
\begin{split}
E_V^{*}  & =  \left(S^{***}, V^{***}, I^{***} \right) \\
 & = \left(\dfrac{(H + I^{***}) [(1-p) \Lambda [(\mu + \omega)(H + I^{***}) + (1-\phi)\beta I^{***}] + \omega p \Lambda (H + I^{***}) ]}{[\mu (H + I) + \beta I] [(\mu + \omega) (H + I^{***}) + (1-\phi) \beta I^{***}]}, \right. \\
 & \quad \left. \dfrac{p \Lambda (H + I^{***})}{(\mu + \omega)(H + I^{***}) + (1-\phi) \beta I^{***}}, \quad \dfrac{-b \pm \sqrt{b^2 -4ac}}{2a} \right)
\end{split}
\end{equation*}

such that 

\begin{equation*}
%\label{eq:vac abc 1}
\begin{aligned}
     a & = - (\mu + \delta) [\mu \beta(1 - \phi) + (\mu + \omega) (\mu + \beta)+ \beta^2 (1-\phi)], \\
    b &  = \beta^2 \Lambda (1-\phi) + \mu H (\mu + \delta) (\mu + \omega) (\mathcal{R}_V -1) \\
    & \qquad - (\mu + \delta) H [\mu \beta (1-\phi) + (\mu + \beta) (\mu + \omega)], \\
    c &  = \mu H^2 (\mu + \delta) (\mu + \omega) (\mathcal{R}_V -1). 
\end{aligned}
\end{equation*}

The vaccination model (\ref{eq:vac model}) has no endemic equilibrium when $\mathcal{R}_V \leq 1$, and has a unique endemic equilibrium when $\mathcal{R}_V > 1$. Fig.~\ref{fig:bifurcation}{C} illustrates a bifurcation plot between the population of infected birds and the basic reproduction number $\mathcal{R}_V$, showing a forward bifurcation. This bifurcation diagram is in line with our result in Theorem \ref{thm:vac endemic eq}, so there is no endemic equilibrium when $\mathcal{R}_V < 1$, but there is a unique endemic equilibrium when $\mathcal{R}_V > 1$. In this case, reducing $\mathcal{R}_V$ below one is sufficient to control the disease.

Finally, we analyze the stability of equilibria of the AIV model with culling (\ref{eq:cul model}). The DFE for the culling model is given by 
\begin{equation*}
 \label{eq:culDFE}
    E_C^0 = \left(S^0, I^0 \right) = \left(\dfrac{\Lambda }{\mu}, 0\right).
 \end{equation*} 
and the basic reproduction number is  
 \begin{align*}
\label{eq:cul RC}
   \mathcal{R}_C= \dfrac{\beta \Lambda}{H \mu [\mu + \delta]}.
\end{align*}

The endemic equilibria of the culling model is determined as 

 \begin{equation}
 \label{eq:cul endemic}
\begin{split}
E_C^{*}  & =  \left(S^{****}, I^{****} \right) = \left(\dfrac{\Lambda (H +I^{****} )}{\mu H + (\mu + c_s  + \beta) I^{****} }, \dfrac{-b \pm \sqrt{b^2 -4ac}}{2a} \right),
\end{split}
\end{equation}

such that

\begin{equation*}
%\label{eq:cul abc1}
\begin{aligned}
    & a = - (\mu + \delta + c_i) (\mu + c_s + \beta), \\
    & b = \mu H (\mu + \delta)(\mathcal{R}_C - 1) - c_i \mu H - H (\mu + \delta) (\mu + c_s + \beta), \\
    & c = \mu H^2 (\mu + \delta)(\mathcal{R}_C - 1). 
\end{aligned}
\end{equation*}

For the culling model (\ref{eq:cul model}), we have shown that a backward bifurcation does not exist when $\mathcal{R}_C<1$. The culling model (\ref{eq:cul model}) has no endemic equilibrium when $\mathcal{R}_C < 1$, and has a unique endemic equilibrium when $\mathcal{R}_C > 1 $.

In Fig.~\ref{fig:bifurcation}{D}, we have a bifurcation diagram showing the infected population and the basic reproduction number ($\mathcal{R}_C$). We have a forward bifurcation in the plot, which is similar to the result stated in Theorem \ref{thm:cul endemic eq}, implying that, when $\mathcal{R}_C<1$,  avian influenza dies out from the poultry population.

%---------------------------------------------------------

\section{Optimal control strategies}
%---------------------------------------------------------

We now integrate an optimal-control approach in all our models: isolation, vaccination, and culling.
\subsection{Isolation}
Our first control involves isolating infected birds with $u_1$ replacing $\psi$. The second control indicates the effort of the farmers in choosing a drug that can increase the success of treatment with $u_2$ replacing $f$. The isolation model (\ref{eq:iso model}) becomes
\begin{equation}
\label{eq:opt iso model}
\begin{split}
        \dot{S} &= \Lambda - \mu S - \dfrac{\beta S I}{H + I} , \\
        \dot{I} &=\dfrac{\beta S I}{H + I} + \left(1- {{u_2(t)}} \right)\gamma T - \left(\mu + \delta + {{u_1(t)}} \right)I , \\
        \dot{T} &= {{u_1(t)}} I - (\mu + \delta + \gamma)T, \\
        \dot{R} &= {{u_2 (t)}} \gamma T - \mu R.
\end{split}
\end{equation}

We represent the rate of isolation of infected birds by control $u_1(t)$, that is the rate $u_1(t) I$ transfers from $I$ to $T$. The proportion of successfully treated birds released from isolation is denoted by $u_2(t)$.

The problem is to minimize the objective functional defined by

\begin{equation*}
\label{eq:iso_cost}
    J_I(u_1, u_2) = \int_0^{t_f} \left[ I(t) + T(t) + \frac{B_1}{2} u_1^2(t) + \frac{B_2}{2} u_2^2 (t) \right]dt,
\end{equation*}
 which is subject to the ordinary differential equations in (\ref{eq:opt iso model}) and where $t_f$ is the final time. The objective functional includes isolation control ($u_1(t)$) and treatment control ($u_2(t)$),  while $B_1$ and  $B_2$ are weight constants associated to relative costs of applying respective control strategy. Given that we have two controls $u_1(t)$ and $u_2(t)$, we want to find the optimal controls $u_1^*(t)$  and $u_2^*(t)$ such that 
 
 \begin{equation*}
 \label{eq:iso_cost_min}
     J_I(u_1^*, u_2^*) = \min_{\mathcal{U}_I} \{J_I(u_1, u_2)\},
 \end{equation*} 
 where $\mathcal{U}_I = \{(u_1, u_2)| \, u_i: [0, t_f] \rightarrow [a_i, b_i],i = 1,2,\text{ is Lebesgue integrable} \}$ is the control set. We consider the worst and best scenarios of isolating infected birds and giving treatment by letting the lower bounds $a_i = 0$ and upper bounds $b_i = 1$, for $i=1,2$.

%---------------------------------------------------------
\subsubsection{Characterization of optimal control for isolation strategy}
%---------------------------------------------------------

We generate the necessary conditions of this optimal control using Pontryagin's Maximum Principle \cite{pontryagin_mathematical_1986}. We define the Hamiltonian, denoted by $H_I$, as follows:
\begin{equation}
\label{eq:iso_Hamiltonian}
    \begin{split}
        H_I & = I(t) + T(t) + \frac{B_1}{2} u_1^2 (t) + \frac{B_2}{2} u_2^2 (t) + \lambda_{{I}_1} \left(\Lambda - \mu S - \frac{\beta S I}{H + I} \right) \\
         & \quad + \lambda_{{I}_2} \left( \frac{\beta S I}{H + I} + \left[1-u_2(t) \right] \gamma T - \left[\mu + \delta + u_1(t) \right]I\right) \\
         & \quad + \lambda_{{I}_3}  \left( u_1 (t) I - (\mu + \delta + \gamma) T \right) + \lambda_{{I}_4} \left( u_2(t) \gamma T - \mu R \right),
    \end{split}
\end{equation} 
where $\lambda_{{I}_1}, \lambda_{{I}_2}, \lambda_{{I}_3}, \lambda_{{I}_4}$ are the associated adjoints  for the states $S, I, T, R$. We obtain the system of adjoint equations by using the partial derivatives of the Hamiltonian (\ref{eq:iso_Hamiltonian}) with respect to each state variable.

\begin{theorem}
There exists optimal controls $u_1^*(t)$ and $u_2^*(t)$ and solutions $S^*, I^*, T^*, R^*$ of the corresponding state system (\ref{eq:opt iso model}) that minimizes the objective functional $J_I (u_1(t), u_2(t))$ over $\mathcal{U}_I$. Then there exists adjoint variables $\lambda_{I_1}, \lambda_{{I}_2},\lambda_{{I}_3},$ and $\lambda_{{I}_4}$ satisfying

\begin{equation*}
\label{eq:iso_adjoint}
\begin{split}
   \dfrac{d\lambda_{{I}_1}}{dt} &= \lambda_{{I}_1} \left(\mu + \dfrac{\beta I}{H+I}\right) - \lambda_{{I}_2} \left( \dfrac{\beta I}{H+I}\right), \\
    \dfrac{d\lambda_{{I}_2}}{dt}&=  -1 + \lambda_{{I}_1}\left( \dfrac{H \beta S}{(H+I)^2} \right) - \lambda_{{I}_2}\left( \dfrac{H \beta S}{(H+I)^2} \right) + \lambda_{{I}_2} \left[ \mu + \delta + u_1(t) \right] - \lambda_{{I}_3} u_1 (t)\\
   \dfrac{d\lambda_{{I}_3}}{dt} &= -1 - \lambda_{{I}_2} \left[ 1- u_2(t) \right] \gamma + \lambda_{{I}_3}(\mu + \delta + \gamma) - \lambda_{{I}_4} u_2(t) \gamma, \\
   \dfrac{d\lambda_{{I}_4}}{dt} &= \lambda_{{I}_4} \mu
\end{split}
\end{equation*}
with transversality conditions $\lambda_{{I}_i}(t_f) = 0$, for $i = 1,2,3,4$. 
Furthermore,
\begin{equation}
\label{eq:iso_u1u2}
\begin{array}{ll} %
     &  u_1^* = \min \left\{ b_1, \max \left\{ a_1,  \dfrac{ \left(\lambda_{{I}_2}-\lambda_{{I}_3} \right)I}{B_1} \right\} \right\} \quad and \quad u_2^* = \min \left\{ b_2, \max \left\{ a_2,  \dfrac{ \left(\lambda_{{I}_2}-\lambda_{{I}_4} \right) \gamma T}{B_2} \right\} \right\}.
\end{array}
\end{equation}
\end{theorem}

\begin{proof}
The existence of optimal control $(u_1^*, u_2^*)$ is given by the result of Fleming and Rishel (1975). Boundedness of the solution of our system (\ref{eq:iso model}) shows the existence of a solution for the system. We have nonnegative values for the controls and state variables.  In our minimizing problem, we have a convex integrand for $J_I$ with respect to $(u_1, u_2)$. By definition, the control set is closed, convex, and compact which shows the existence of optimality solutions in our optimal system. By Pontryagin's Maximum Principle \cite{pontryagin_mathematical_1986}, we obtain the adjoint equations and transversality conditions. We differentiate the Hamiltonian (\ref{eq:iso_Hamiltonian}) with respect to the corresponding state variables as follows:

\begin{equation*}
       \dfrac{d\lambda_{I_1}}{dt} = - \dfrac{\partial H_I}{\partial S} , \quad
       \dfrac{d\lambda_{I_2}}{dt} = -\dfrac{\partial H_I}{\partial I}, \quad
       \dfrac{d\lambda_{I_3}}{dt} = - \dfrac{\partial H_I}{\partial T}, \quad
       \dfrac{d\lambda_{I_4}}{dt} = - \dfrac{\partial H_I}{\partial R} 
\end{equation*} 
with $\lambda_{I_i}(t_f)=0$ where $i = 1, 2, 3, 4$. We consider the optimality condition 

\begin{equation*}
\begin{array}{ll}
     & \dfrac{\partial H_I}{\partial u_1} = B_1 u_1(t) - \lambda_{I_2} I + \lambda_{I_3} I = 0 \quad \text{and} \quad \dfrac{\partial H_I}{\partial u_2} = B_2 u_2(t) - \lambda_{I_2} \gamma T + \lambda_{I_4} \gamma T = 0
\end{array}
\end{equation*}

to derive the optimal controls in (\ref{eq:iso_u1u2}). We consider the bounds of the controls and obtain the characterization for optimal controls $u_1^*$ and $u_2^*$ as follows:

\begin{equation*}
\begin{array}{ll} %
     &  u_1^* = \min \left\{ 1, \max \left\{ 0,  \dfrac{ \left(\lambda_{\mathcal{I}_2}-\lambda_{\mathcal{I}_3} \right)I}{B_1} \right\} \right\} \quad and \quad u_2^* = \min \left\{ 1, \max \left\{ 0,  \dfrac{ \left(\lambda_{\mathcal{I}_2}-\lambda_{\mathcal{I}_4} \right) \gamma T}{B_2} \right\} \right\}.
\end{array}
\end{equation*}
\end{proof}

\subsection{Vaccination}
%-------------------------------------------------------

For vaccination, the first control represents the effort of the farmers to increase vaccinated birds, while the other control describes the efficacy of the vaccine in providing immunity against H5N6. Where $u_3(t)$ and $u_4(t)$ replace $p$ and $\phi$, respectively, into the vaccination model (\ref{eq:vac model}) to obtain
\begin{equation}
\label{eq:opt vac model}
\begin{split}
   \dot{S} &= \left(1- {{u_3(t)}} \right) \Lambda + \omega V - \mu S - \dfrac{\beta S I}{H + I}, \\
\dot{V} &= {{u_3(t)}} \Lambda - (\mu + \omega) V - [1-u_4(t)] \dfrac{\beta VI}{H + I}, \\
\dot{I} &=\dfrac{\beta S I}{H + I} + [1-u_4(t)] \dfrac{\beta V I}{H + I} - (\mu + \delta)I.
\end{split}
\end{equation}

We describe the proportion of birds that are vaccinated by the control $u_3(t)$ and the immunity of the vaccinated population against acquiring the disease  by $u_4(t)$.  

We have the objective functional
\begin{equation*}
\label{eq:vac_cost}
    J_V(u_3, u_4) = \int_0^{t_f} \left[ I(t) + \frac{B_3}{2} u_3^2(t) + \frac{B_4}{2} u_4^2 (t) \right]dt,
\end{equation*}
which is subject to (\ref{eq:vac model}). This objective functional involves increased vaccination $u_3(t)$ and the vaccine-efficacy control $u_4(t)$, where $B_3$ and $B_4$ are the weight constants representing the relative cost of implementing each respective controls. We need to find the optimal controls $u_3^*(t)$ and $u_4^*(t)$ such that 
 \begin{equation*}
 \label{eq:vac_cost_min}
     J_V(u_3^*, u_4^*) = \min_{\mathcal{U}_V} \{J_V(u_3, u_4)\},
 \end{equation*} 
 where $\mathcal{U}_V = \left\{ (u_3, u_4) | u_i :[0, t_f] \rightarrow [a_i, b_i], \, i = 3,4, \,\text{ is Lebesgue integrable}\right\}$ is the control set. We consider the lower bound $a_{i}=0$ and upper bounds $b_{i}=1$, for $i= 3, 4$.

%---------------------------------------------------------
\subsubsection{Characterization of optimal control for vaccination strategy}
%---------------------------------------------------------

Similarly, we use Pontryagin's Maximum Principle \cite{pontryagin_mathematical_1986} to show necessary conditions of optimal control. We define the Hamiltonian denoted by $H_V$ as follows:
\begin{equation}
\label{eq:vac_Hamiltonian}
    \begin{split}
        H_V & = I(t) + \frac{B_3}{2} u_3^2 (t) + \frac{B_4}{2} u_4^2 (t) + \lambda_{{V}_1} \left[ (1-u_3(t))\Lambda + \omega V - \mu S - \frac{\beta S I}{H + I} \right] \\
         & \quad + \lambda_{{V}_2} \left( u_3(t) \Lambda - (\mu + \omega)V - (1- u_4(t))\frac{\beta V I}{H + I} \right)  \\
         & \quad + \lambda_{{V}_3} \left( \frac{\beta S I}{H + I} + [1- u_4(t)]\frac{\beta V I}{H + I} - (\mu + \delta) I\right). 
    \end{split}
\end{equation}
 
\begin{theorem}
There exists optimal controls  $u_3^*(t)$ and $u_4^*(t)$ and solutions $S^*, V^*, I^*$ of the corresponding state system (\ref{eq:opt vac model}) that minimize the objective functional $J_V (u_3(t), u_4(t))$ over $\mathcal{U}_V$. Then there exists adjoint variables $\lambda_{V_1}, \lambda_{{V}_2}$ and $\lambda_{{V}_3}$ satisfying
\begin{equation*}
\label{eq:vac_adjoint}
\begin{split}
   \dfrac{d\lambda_{{V}_1}}{dt} ={}& \lambda_{{V}_1} \left(\mu + \dfrac{\beta I}{H+I}\right) - \lambda_{{V}_3} \left( \dfrac{\beta I}{H+I}\right), \\    
    \dfrac{d\lambda_{{V}_2}}{dt}={}& - \lambda_{{V}_1} \omega + \lambda_{{V}_2}\left( \mu + \omega + [1-u_4(t)]\dfrac{ \beta I}{H+I} \right) - \lambda_{{V}_3} \left[ 1- u_4(t) \right] \dfrac{\beta I}{H+I},\\ 
    \dfrac{d\lambda_{{V}_3}}{dt}={}&  -1 + \lambda_{{V}_1}\left[ \dfrac{H \beta S}{(H+I)^2} \right] + \lambda_{{V}_2}\left[ [1-u_4(t)]\dfrac{H \beta V}{(H+I)^2} \right] \\
    \qquad \qquad& - \lambda_{{V}_3} \left[ \dfrac{H \beta S}{(H + I)^2} + [1-u_4(t)] \dfrac{H \beta V}{(H + I)^2} - (\mu + \delta) \right], 
\end{split}
\end{equation*} 
with transversality conditions $\lambda_{{V}_i}(t_f) = 0$, for $i = 1,2,3$. Furthermore, 
\begin{equation}
\label{eq:vac_u3u4}
\begin{array}{ll}
     &   u_3^* = \min \left\{ b_3, \max \left\{ a_3,  \dfrac{ \left(\lambda_{{V}_1}-\lambda_{{V}_2} \right) \Lambda}{B_3} \right\} \right\} \quad and \quad  u_4^* = \min \left\{ b_4, \max \left\{ a_4,  \dfrac{ \left(\lambda_{{V}_3}-\lambda_{{V}_2} \right) \beta V I}{B_4 (H + I)} \right\} \right\}.
\end{array}
\end{equation}
\end{theorem}

\begin{proof}
Similarly, the existence of optimal control $(u_3^*, u_4^*)$ is given by the result of Fleming and Rishel (1975). Boundedness of the solution of our system (\ref{eq:vac model}) shows the existence of a solution for the system. We have nonnegative values for the controls and state variables.  In our minimizing problem, we have a convex integrand for $J_V$ with respect to $(u_3, u_4)$. By definition, the control set is closed, convex, and compact which shows the existence of optimality solutions in our optimal system. We use Pontryagin's Maximum Principle \cite{pontryagin_mathematical_1986} to obtain the adjoint equations and transversality conditions. We differentiate the Hamiltonian (\ref{eq:vac_Hamiltonian}) with respect to the corresponding state variables as follows:

\begin{equation*}
       \dfrac{d\lambda_{V_1}}{dt} = - \dfrac{\partial H_V}{\partial S} , \quad
       \dfrac{d\lambda_{V_2}}{dt} = -\dfrac{\partial H_V}{\partial V}, \quad
       \dfrac{d\lambda_{V_3}}{dt} = - \dfrac{\partial H_V}{\partial I}
\end{equation*} 
with $\lambda_{V_i}(t_f)=0$ where $i = 1, 2, 3$. Using the optimality condition 

\begin{equation*}
\begin{array}{ll}
     & \dfrac{\partial H_V}{\partial u_3} = B_3 u_3(t) - \lambda_{V_1} \Lambda + \lambda_{V_2} \Lambda = 0 \quad and \quad \dfrac{\partial H_V}{\partial u_4} = B_4 u_4(t) + \lambda_{V_2} \dfrac{\beta V I}{H + I} - \lambda_{V_3} \dfrac{\beta V I}{H+I} = 0.
\end{array}
\end{equation*}

we derive the optimal controls (\ref{eq:vac_u3u4}). We consider the bounds for the control and conclude the characterization for $u_3^*$ and $u_4^*$

\begin{equation*}
\begin{array}{ll}
     &   u_3^* = \min \left\{ 1, \max \left\{0,  \dfrac{ \left(\lambda_{{V}_1}-\lambda_{{V}_2} \right) \Lambda}{B_3} \right\} \right\} \quad and \quad  u_4^* = \min \left\{ 1, \max \left\{ 0,  \dfrac{ \left(\lambda_{{V}_3}-\lambda_{{V}_2} \right) \beta V I}{B_4 (H + I)} \right\} \right\}.
\end{array}
\end{equation*}

\end{proof}

%-------------------------------------------------------
\subsection{Culling}
%-------------------------------------------------------

Finally, we administer optimal control to the culling model (\ref{eq:cul model}). Thus we have
\begin{equation}
\label{eq:opt cul model}
\begin{split}
        \dot{S} &= \Lambda - \mu S - \dfrac{u_5(t)SI}{H + I} - \dfrac{\beta S I}{H + I}, \\
        \dot{I} &=\dfrac{\beta S I}{H + I} - (\mu + \delta)I - \dfrac{{{u_6(t)}} I^2}{H + I}.
\end{split}
\end{equation}

We represent the frequency of culling the susceptible population by $u_5(t)$ and frequency of culling the infected population by $u_6(t)$. We have the objective functional 
\begin{equation*}
\label{eq:cul_cost}
    J_C(u_5, u_6) = \int_0^{t_f} \left[ I(t) + \frac{B_5}{2} u_5^2(t) + \frac{B_6}{2} u_6^2 (t) \right]dt,
\end{equation*} 
which is subject to (\ref{eq:cul model}). The objective functional includes the susceptible and infected culling control denoted by $u_5(t)$ and $u_6(t)$, respectively, with $B_5$ and $B_6$ as the weight constants representing the relative cost of implementing each respective controls. Hence we have to find the optimal controls $u_5^*$ and $u_6^*$ such that 

 \begin{equation*}
 \label{eq:cul_cost_min}
     J_C(u_5^*, u_6^*) = \min_{\mathcal{U}_V} \{J_C(u_5, u_6)\},
 \end{equation*} 
 where $\mathcal{U}_C = \left\{ (u_5, u_6) | [0, t_f] \rightarrow [a_i, b_i], \, i = 5,6, \,\text{ is Lebesgue integrable}\right\}$ is the control set. We consider the lower bound $a_{i}=0$ and upper bounds $b_{i}=1$, for $i= 5, 6.$

%---------------------------------------------------------
\subsubsection{Characterization of optimal control for culling strategy}
%---------------------------------------------------------
 
We utilize the Pontryagin's Maximum Principle \cite{pontryagin_mathematical_1986} to show the necessary conditions of optimal control. We define the Hamiltonian denoted by $H_C$ as follows:
\begin{equation}
\label{eq:cul_Hamiltonian}
    \begin{split}
        H_C & = I(t) + \frac{B_5}{2} u_5^2 (t) + \frac{B_6}{2} u_6^2 (t) + \lambda_{{C}_1} \left[\Lambda - \mu S - \frac{u_5(t) S I}{H + I} - \frac{\beta S I}{H+I}\right] \\
         & \quad + \lambda_{{C}_2} \left( \frac{\beta S I}{H+I} - (\mu + \delta)I - \frac{u_6(t)I^2}{H + I} \right).
    \end{split}
\end{equation}

\begin{theorem}
There exists optimal controls  $u_5^*(t)$ and $u_6^*(t)$ and solutions $S^*, I^*$ of the corresponding state system (\ref{eq:opt cul model}) that minimize the objective functional $J_C (u_5(t), u_6(t))$ over $\mathcal{U}_C$. Then there exists adjoint variables $\lambda_{C_1}$ and $\lambda_{C_2}$ satisfying
\begin{equation*}
\label{eq:cul_adjoint}
\begin{split}
   \dfrac{d\lambda_{{C}_1}}{dt} &= \lambda_{{C}_1} \left[\mu  +\dfrac{u_5(t) I}{H+I} + \dfrac{\beta I}{H+I}\right] - \lambda_{{C}_2} \dfrac{\beta I}{H+I}, \\ 
       \dfrac{d\lambda_{{C}_2}}{dt}&=  -1 + \lambda_{{C}_1}\left[ \dfrac{u_5(t) HS}{(H+I)^2} + \dfrac{H \beta S}{(H+I)^2} \right] - \lambda_{{C}_2}\left[ \dfrac{H \beta S}{(H+I)^2} - (\mu + \delta) - \dfrac{(2H + I) u_6(t) I}{(H+I)^2} \right] \\
\end{split}
\end{equation*} 
with transversality conditions $\lambda_{C_i}(t_f) = 0$ for $i = 1,2$. Furthermore, 
\begin{equation}
\label{eq:cul_u5u6}
    \begin{array}{ll}
         &  u_5^* = \min \left\{ b_5, \max \left\{ a_5,  \dfrac{ \lambda_{\mathcal{C}_1} SI}{B_5(H+I)} \right\} \right\} \quad and \quad u_6^* = \min \left\{ b_6, \max \left\{ a_6,  \dfrac{ \lambda_{\mathcal{C}_2}I^2}{B_6 (H + I)} \right\} \right\}.
    \end{array}
\end{equation}
\end{theorem}

\begin{proof}
The existence of optimal control $(u_5^*, u_6^*)$ is given by the result of Fleming and Rishel (1975). Boundedness of the solution of our system (\ref{eq:cul model}) shows the existence of a solution for the system. We have nonnegative values for the controls and state variables.  In our minimizing problem, we have a convex integrand for $J_C$ with respect to $(u_5, u_6)$. By definition, the control set is closed, convex, and compact which shows the existence of optimality solutions in our optimal system. By Pontryagin's Maximum Principle \cite{pontryagin_mathematical_1986}, we obtain the adjoint equations and transversality conditions. We differentiate the Hamiltonian (\ref{eq:cul_Hamiltonian}) with respect to the corresponding state variables as follows:

\begin{equation*}
       \dfrac{d\lambda_{C_1}}{dt} = - \dfrac{\partial H_C}{\partial S} \quad and \quad \dfrac{d\lambda_{C_2}}{dt} = -\dfrac{\partial H_C}{\partial I}
\end{equation*} 
with $\lambda_{C_i}(t_f)=0$ where $i = 1, 2$. We consider the optimality condition 

\begin{equation*}
\begin{array}{ll}
     & \dfrac{\partial H_C}{\partial u_5} = B_5 u_5(t) -  \dfrac{\lambda_{C_1}SI}{H+I} = 0 \quad and \quad \dfrac{\partial H_C}{\partial u_6} = B_6 u_6(t) - \dfrac{\lambda_{C_2}I^2}{H+I} = 0,
\end{array}
\end{equation*}

to derive the optimal controls (\ref{eq:cul_u5u6}). We consider the bounds of the controls and get the characterization for $u_5^*$ and $u_6^*$
\begin{equation*}
    \begin{array}{ll}
         &  u_5^* = \min \left\{ 1, \max \left\{ 0,  \dfrac{ \lambda_{\mathcal{C}_1} SI}{B_5(H+I)} \right\} \right\} \quad and \quad u_6^* = \min \left\{ 1, \max \left\{ 0,  \dfrac{ \lambda_{\mathcal{C}_2}I^2}{B_6 (H + I)} \right\} \right\}.
    \end{array}
\end{equation*}
\end{proof}

%-------------------------------------------------------

\section{Numerical simulations}
The parameter values applied to generate our simulations are listed in the table in the appendix. The initial conditions of the simulations are based on  the Philippines' H5N6 outbreak report given by the OIE \cite{noauthor_oie_2018}. We set $S(0) = 407\,\,837$, $I(0) = 73\,\,360$, $T(0) = 0$, $R(0) = 0$, and the total population of birds $N(0) = 481 \, \, 197$.

%-----------------------------------------------------

\begin{figure}[ht]
\begin{center}
    {\includegraphics[width=2.5in]{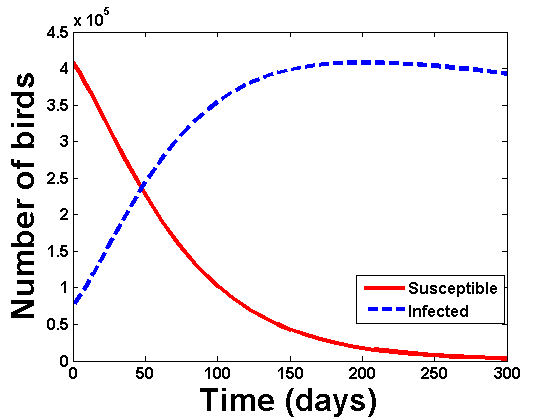}}
\vspace*{8pt}
  \caption{Simulation results showing the transmission dynamics of H5N6 in the Philippines with no intervention strategy. We use initial conditions and parameter values as follows: $S(0) = 407\,\,837$, $I = 73\,\,360$, ${\Lambda = \frac{2\,\,060}{365}}$, $\mu = 3.4246 \times 10^{-4}$, $\beta = 0.025$, $H = 180 \,\,000$, $\delta = 4 \times 10^{-4}$ }
  \label{fig:avian simulation}
\end{center}
\end{figure}
%-----------------------------------------------------
Previous studies suggested that the basic reproduction number for the presence of avian influenza without applying any intervention strategy is $\mathcal{R}_A = 3$ \cite{mills_transmissibility_2004,ward_estimation_2009}. Given this assumption, we have calculated the transmissibility of the disease ($\beta= 0.025$) based on (\ref{eq:avian RA}). Without any control strategy, avian influenza will become endemic in the poultry population as shown in Fig.~\ref{fig:avian simulation}. After 50 days, the population of the infected poultry exceeds that of susceptible poultry, with all birds eventually infected or dead.

%-------------------------------------------------------
\begin{figure}[ht]
\begin{center}
    \includegraphics[width=4in]{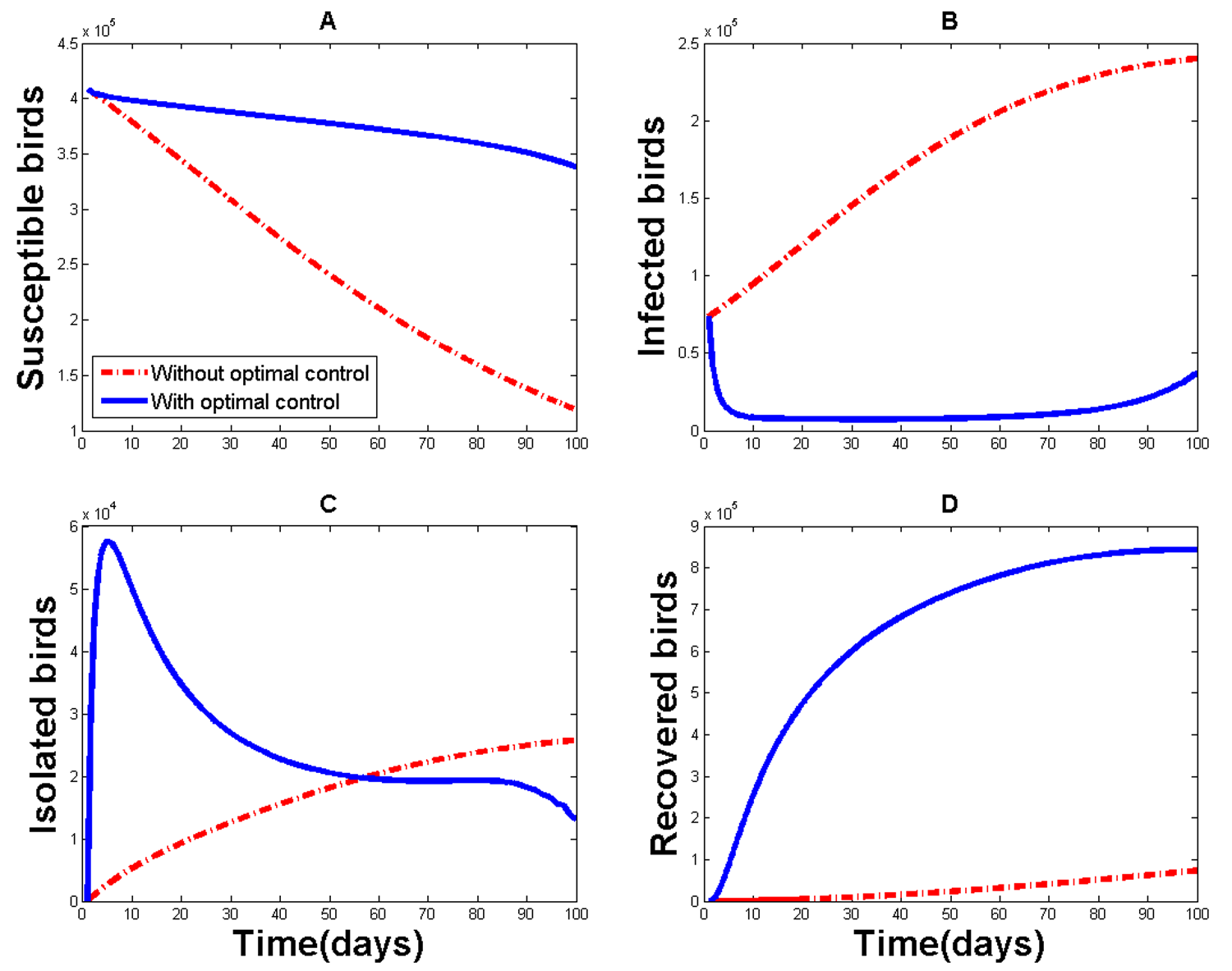}
    \caption{Applying the isolation strategy with (blue solid line) and without (red dashed line) optimal control in the population of susceptible (A), infected (B), isolated (C) and recovered (D) birds.}
    \label{fig:opt_iso_with_without}
\end{center}
    
\end{figure}
%-------------------------------------------------------

Figs.~\ref{fig:opt_iso_with_without}--\ref{fig:opt_iso_both_u1only} illustrate the effects of applying optimal control to isolation strategy under different approaches. These simulations suggest that isolation must be complemented by treatment, where the cheaper cost of implementation works best since it will enable us to apply the strategy in to a larger population of poultry. Application of optimal controls $u_1^*$(t) and $u_2^*$(t) in the susceptible, infected, isolated and recovered population is clearly better than the absence of optimal control (Fig.~\ref{fig:opt_iso_with_without}). We can observe a slower decline of susceptible birds, an initial reduction in infected birds and a delayed increase in infection. More infected birds are isolated (Fig.~\ref{fig:opt_iso_with_without}{C}), and we have a higher number of birds that will recover after going through isolation (Fig. \ref{fig:opt_iso_with_without}{D}). 

% This results can be achieved when $u_1$ lies between $0.10$ and $0.20$ while $u_2$ has the value between $0.30 $ and $0.10$, as illustrated in Figure \ref{fig:opt_iso_with_without}.

%-------------------------------------------------------
\begin{figure}[ht]
   \begin{center}
       \includegraphics[width=4in]{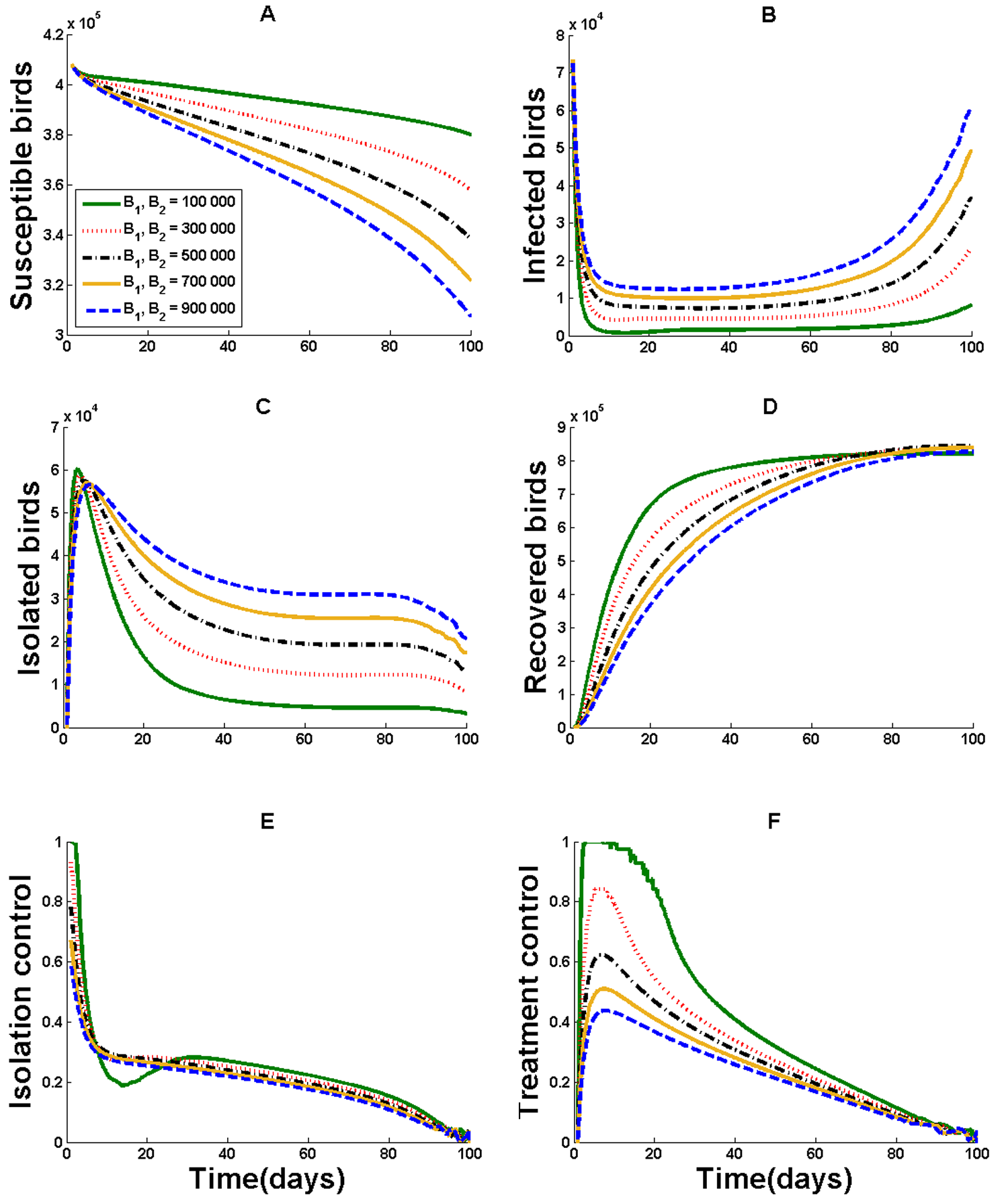}
    \caption{Application of isolation with optimal control to the population of susceptible (A), infected (B), isolated (C) and recovered (D) birds along with isolation control (E) and treatment control (F) for varying values of $B_i$, for i= 1,2, from $100,000$ to $900,000$}
    \label{fig:opt_iso_variedB}
   \end{center}
    
\end{figure}
%-------------------------------------------------------

A cheaper relative cost of implementing both controls $u_1(t)$ and $u_2 (t)$ leads to lower infected populations, as illustrated in Fig.~\ref{fig:opt_iso_variedB}. We can observe that when we have lower values for $B_1$ and $B_2$, the susceptible population has a slower decline, there are fewer infected and isolated birds, and there are more recovered birds. Thus, the cheaper controls are more effective in implementing both isolation and treatment controls.

%-------------------------------------------------------
\begin{figure}[ht]
\begin{center}
    \includegraphics[width=4in]{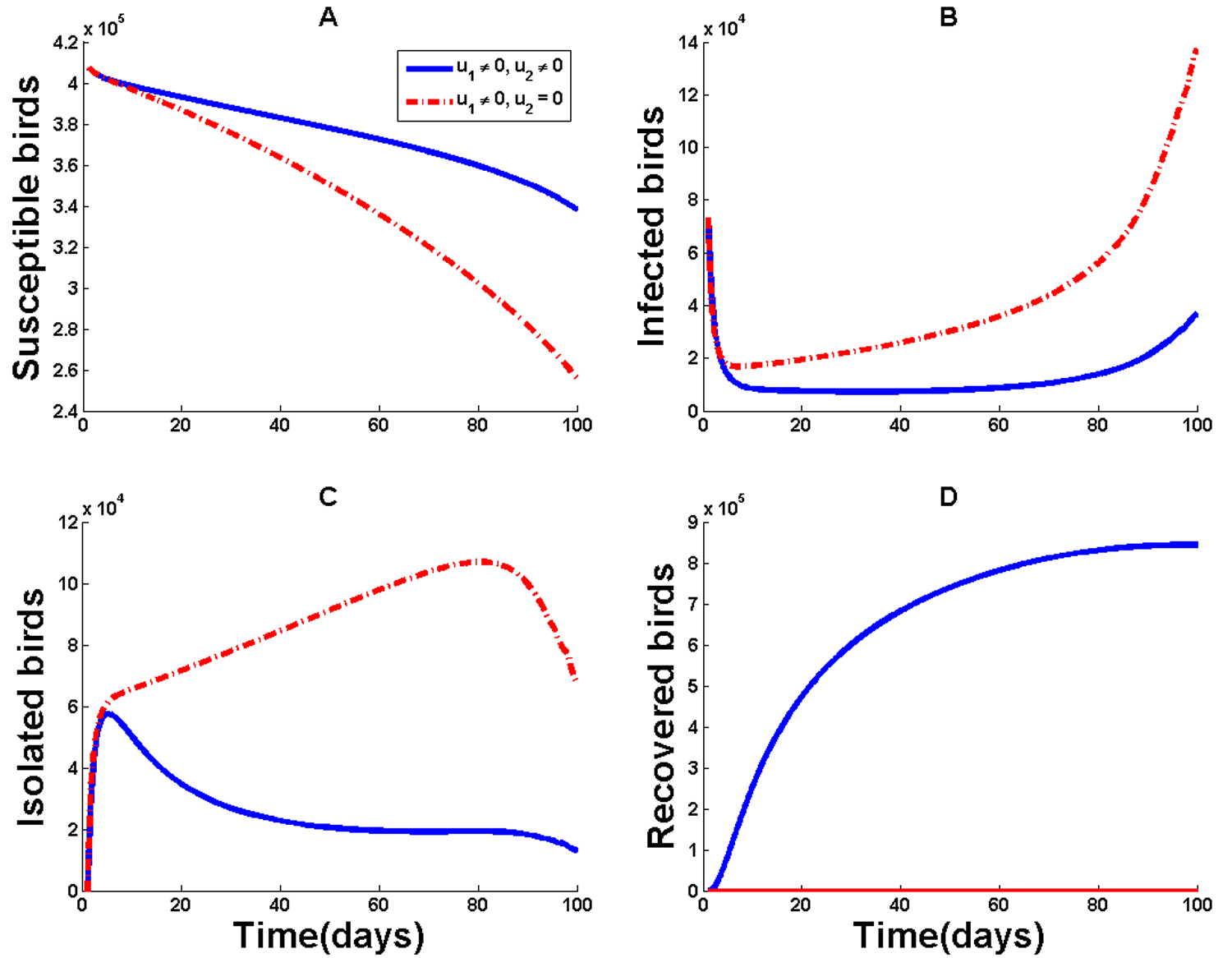}
    \caption{Isolation strategy with the optimal approach and with consideration of using both isolation and treatment control (blue solid line) and using isolation control (red dashed line) only to the population of susceptible (A), infected (B), isolated (C) and recovered (D) birds.}
    \label{fig:opt_iso_both_u1only}
\end{center}
    
\end{figure}
%-------------------------------------------------------

It is evident that using isolation together with treatment showed better results in all populations compared to implementing isolation alone, as depicted in Fig.~\ref{fig:opt_iso_both_u1only}. In applying both controls, the susceptible populations decrease slowly; infected birds are eliminated from the poultry population; and isolated birds increase within 5 days, then decrease afterward. This is due to the release of the birds and the effect of treatment where most of the isolated birds are transferred to the recovered population. Without  treatment isolated birds increase continuously then decrease after 85 days, as illustrated in Fig.~\ref{fig:opt_iso_both_u1only}{C}. The birds were released from isolation zone even though they are still infectious. Our results suggest that the isolation strategy can be maximized by administering isolation together with treatment.

Empirically, we have found that, through the application of optimal control to isolation with treatment strategy, it is possible to control an outbreak, as shown in the numerical simulation from Figs.~\ref{fig:opt_iso_with_without}--\ref{fig:opt_iso_both_u1only}. This also suggests that isolation is more effective if utilized  together with treatment. In addition, a cheaper cost of applying both isolation control and treatment control will result in a lower infected population and more recovered birds.

%-------------------------------------------------------
\begin{figure}[ht]
    \begin{center}
    \includegraphics[width=5in]{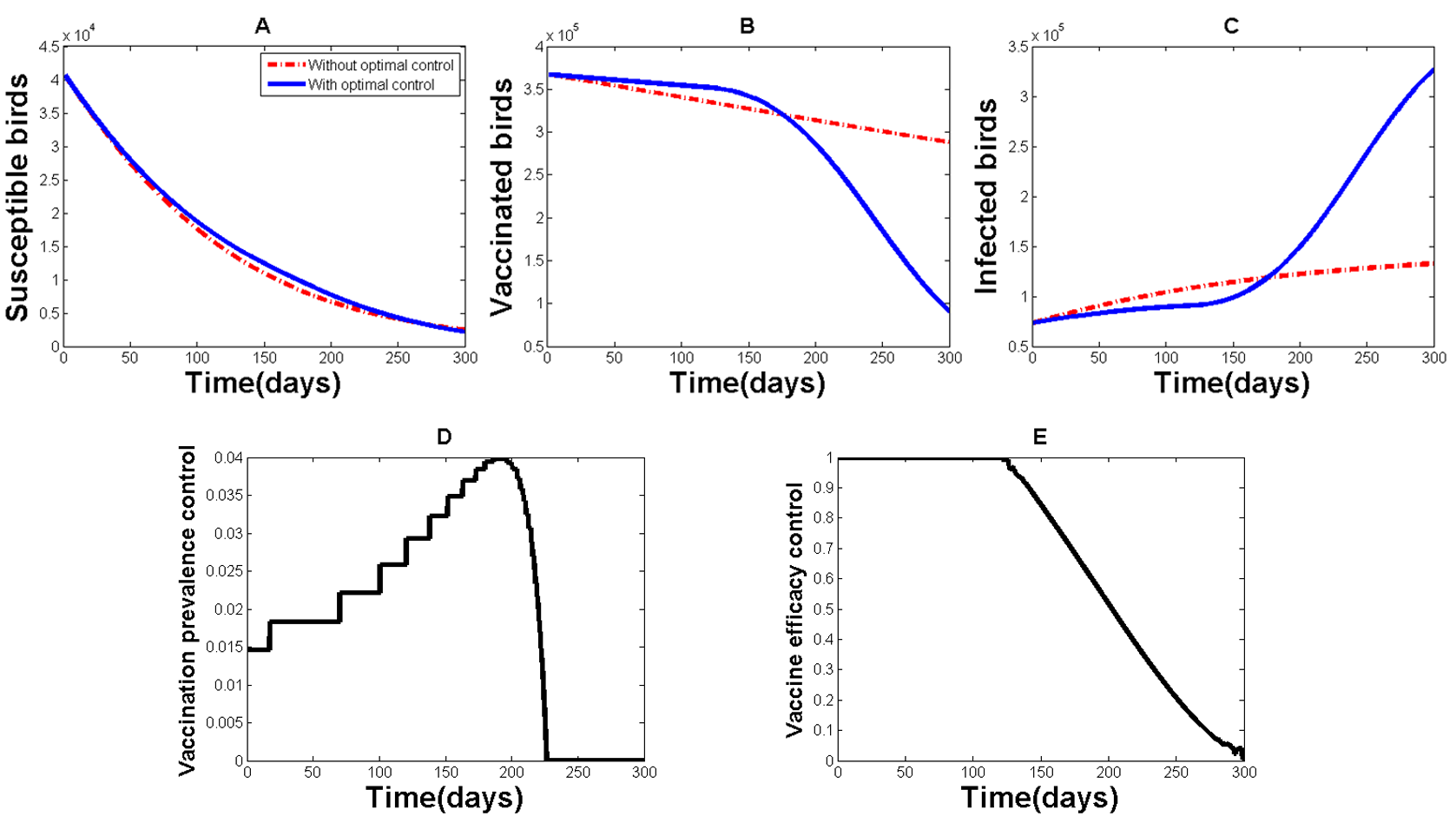}
    \caption{Applying the vaccination strategy with (blue solid line) and without (red dashed line) optimal control in the population of susceptible (A), vaccinated (B) and infected (C) birds, together with the respective values of the increased vaccination (D) and vaccine-efficacy control (E) over time.}
    \label{fig:opt_vac_with_without}
    \end{center}
\end{figure}
%-------------------------------------------------------

Through the application of optimal-control approach in vaccination, we can observe that the diminishing effectiveness of the vaccine results to spread of infection in the vaccinated population, as depicted in Fig.~\ref{fig:opt_vac_with_without}. After 150 days, the vaccine efficacy started to decline causing vaccinated birds to acquire the disease. While in Fig.~\ref{fig:opt_vac_variedB}, taking a lower value for both $B_3$ and $B_4$ provides a higher vaccine efficacy resulting to a higher susceptible and vaccinated population. Hence, for using vaccination strategy, we need to consider cheap vaccine that sustains its effectiveness in a longer period.

%-------------------------------------------------------
\begin{figure}[ht]
\begin{center}
     \includegraphics[width=5in]{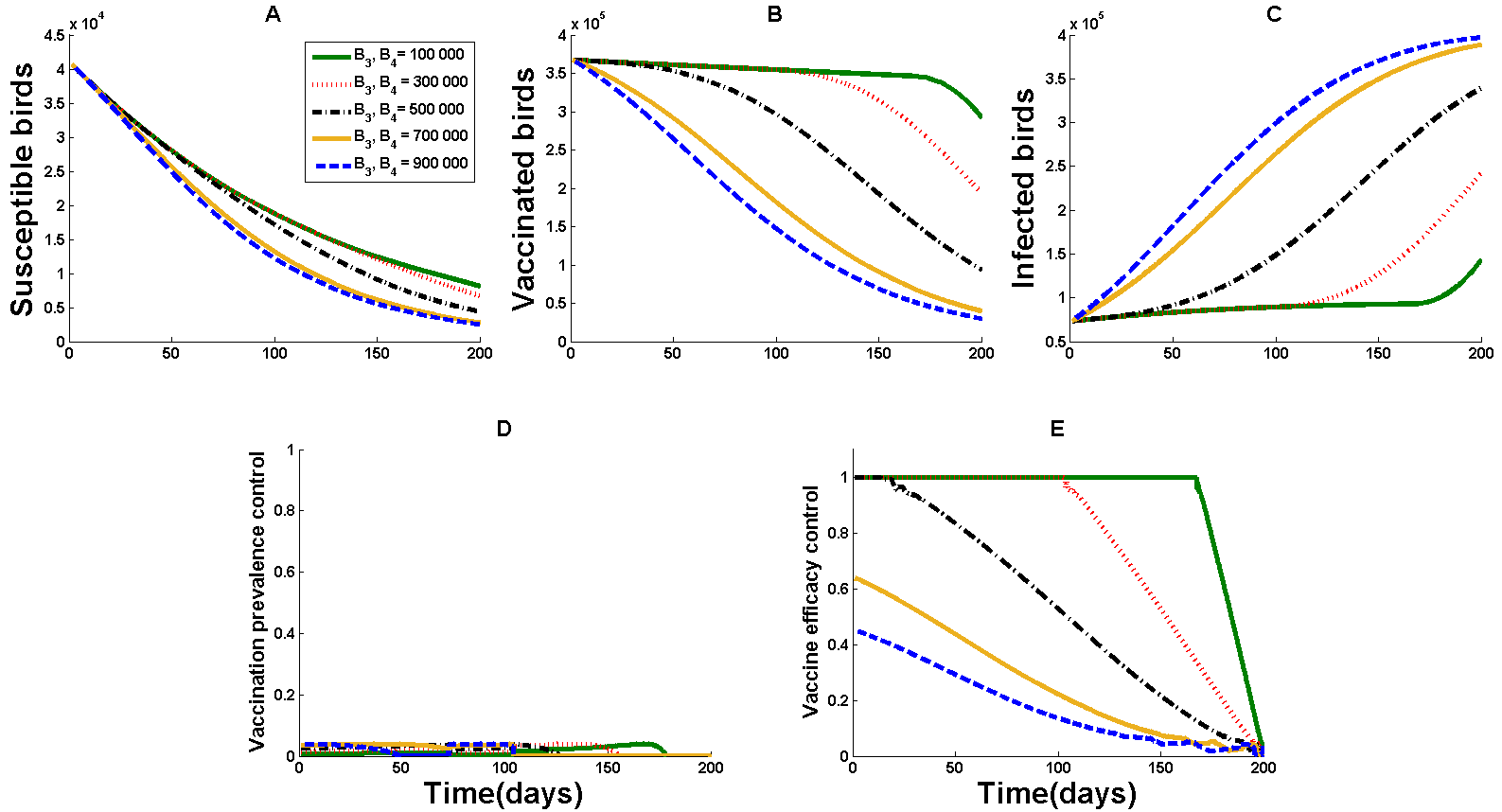}
    \caption{Application of vaccination strategy with optimal control to the population of susceptible (A), vaccinated (B) and infected (C) birds and the increased vaccination (D) and the vaccine-efficacy control (E) with varying values of $B_i$, for i= 3,4, from $100,000$ to $900,000$.}
    \label{fig:opt_vac_variedB}
\end{center}
   
\end{figure}
%-------------------------------------------------------

Simulations shown in Figs.~\ref{fig:opt_vac_with_without}--\ref{fig:opt_vac_variedB} contribute to our understanding that providing immunity to the poultry population is not sufficient to prevent an outbreak due to the possibility of the vaccine to lose its effectiveness. In using an optimal-control approach, we see that a successful immunization strategy highly depends on choosing a long-lasting and an effective vaccine.

Integrating optimal control into a culling strategy results in a lower population for both susceptible and infected birds as compared to using fixed control, as portrayed in Fig. \ref{fig:opt_cul_with_without}. We notice that the decline in the numbers of both susceptible and infected birds occurs faster when optimal control is applied. Culling strategy with cheaper implementation cost results to a  lesser infected population while susceptible population will be in the same level regardless of the implementation cost, as illustrated in Fig.~\ref{fig:opt_cul_variedB}.

%-------------------------------------------------------
\begin{figure}[ht]
  \begin{center}
      \includegraphics[width=3.5in]{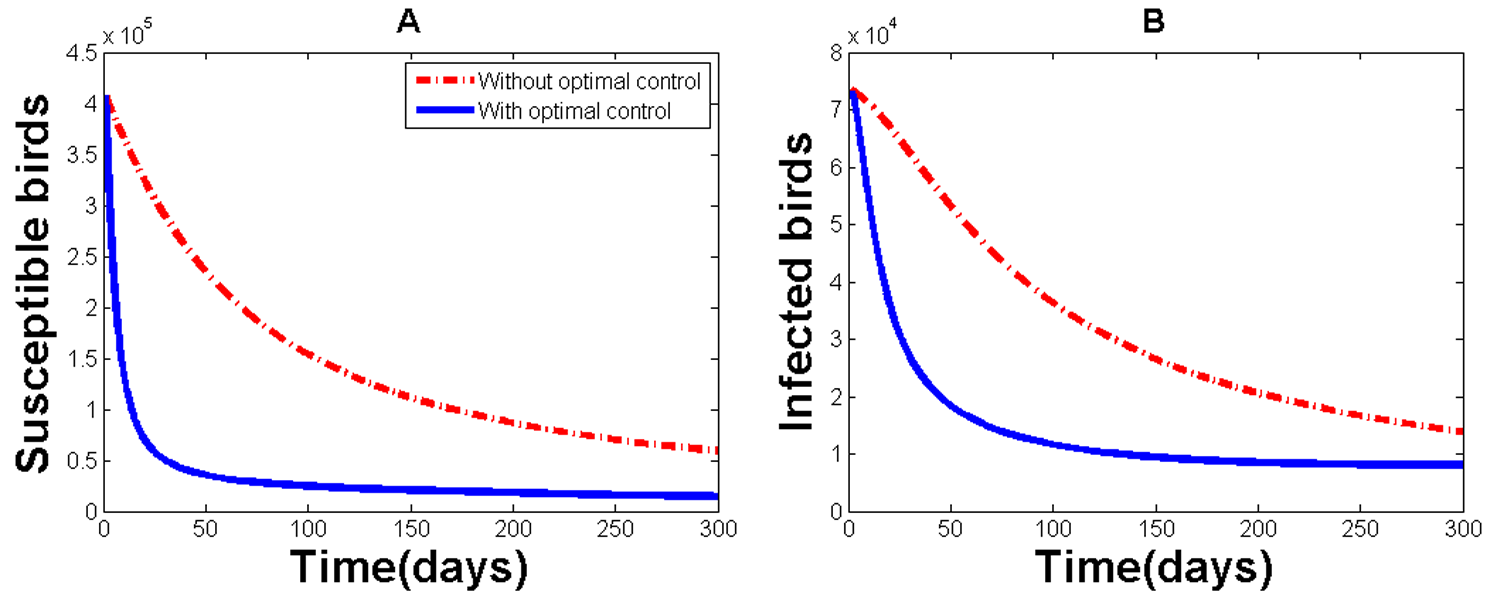}
    \caption{Implementing the culling strategy with optimal control (blue solid line) and without optimal control (red dashed line) in the population of susceptible (A) and infected (B) birds.}
    \label{fig:opt_cul_with_without}
  \end{center}
\end{figure}
%-------------------------------------------------------

%-------------------------------------------------------
\begin{figure}[ht]
 \begin{center}
     \includegraphics[width=3.5in]{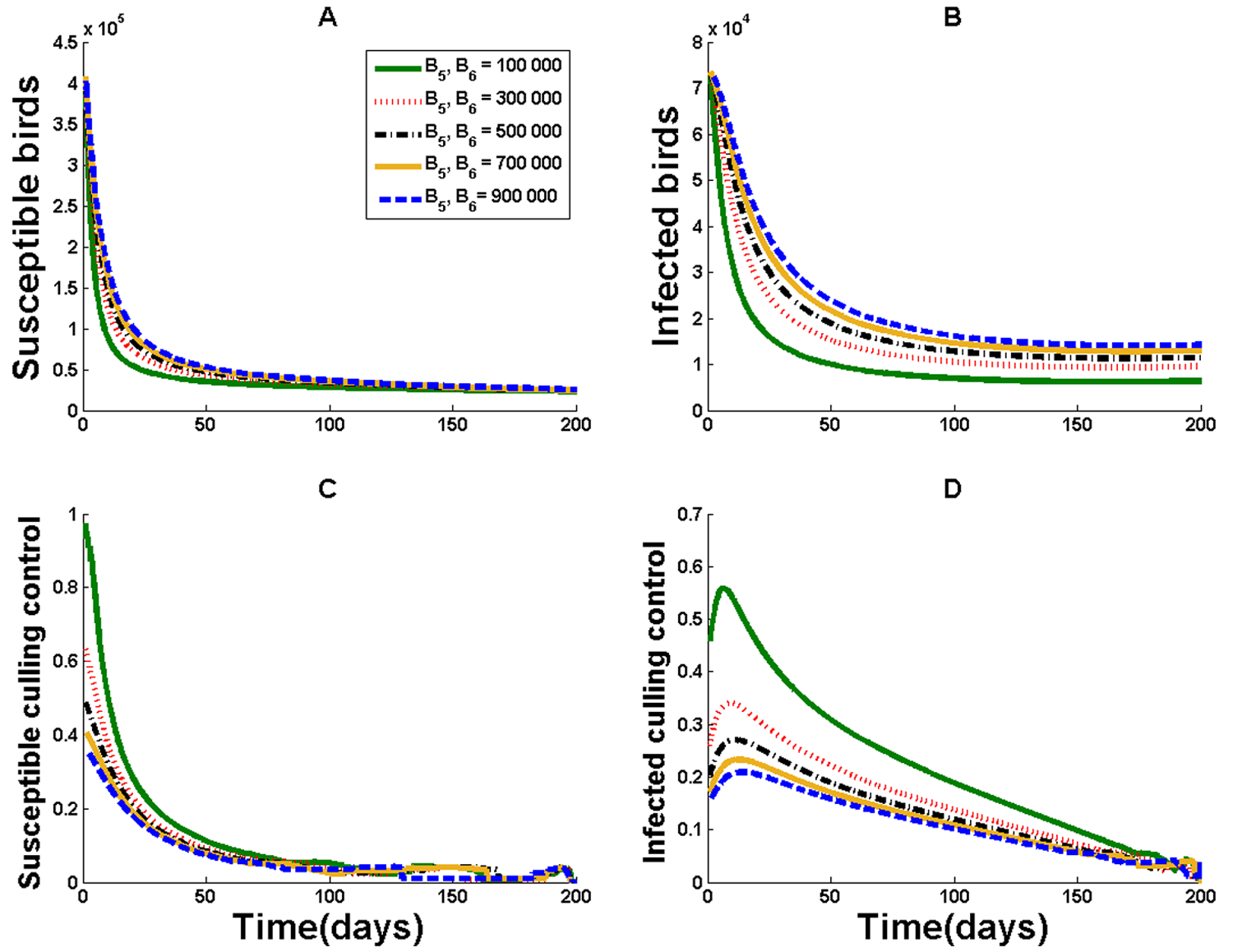}
    \caption{Application of culling strategy with optimal control  to the population of susceptible (A) and infected (B) birds and susceptible culling control (C) and infected culling control (D) with varying values of $B_i$, for i= 5,6, from $100,000$ to $900,000$.}
    \label{fig:opt_cul_variedB}
 \end{center}
    
\end{figure}
%-------------------------------------------------------

Fig.~\ref{fig:opt_cul_variedB}C--D suggests that high culling frequency for both susceptible and infected populations are needed in order to prevent an outbreak. Culling frequency for susceptible birds must be at least 0.30 per day or three times for the first 10 days of the outbreak. The culling frequency for infected birds must be around 0.15--0.6 per day or 2--6 times for the first 10 days and must stay at 0.1--0.3 per day to keep the number of infected birds low.

%-------------------------------------------------------
\begin{figure}[ht]
 \begin{center}
     \includegraphics[width=3.5in]{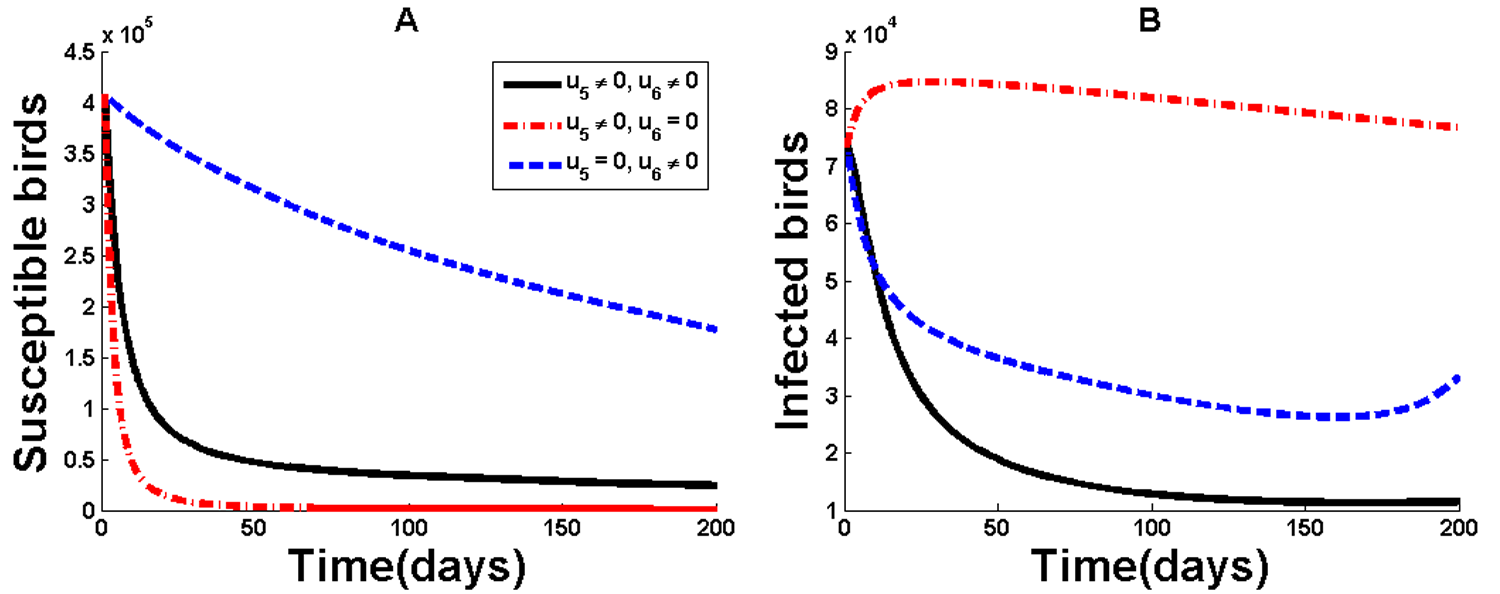}
    \caption{Simulation of culling strategy with the optimal approach and with consideration of using both susceptible culling control $u_5(t)$ and infected culling control $u_6(t)$ (black solid line), using susceptible culling control $u_5(t)$ only (red dotted-dashed line), and using infected culling control $u_6(t)$ to the population of susceptible (A) and infected (B) birds.}
    \label{fig:opt_cul_u67}
 \end{center}
    
\end{figure}
%-------------------------------------------------------

Administering a culling strategy for both susceptible and infected birds is more effective than culling only the infected birds, as indicated in Fig.~\ref{fig:opt_cul_u67}. Looking at the blue dashed line of Fig.~\ref{fig:opt_cul_u67}{A}, we have more susceptible birds if we cull only the infected population, but, as shown in Fig.~\ref{fig:opt_cul_u67}{B}, the number of infected birds increases afterward. This implies that culling only the infected population is not enough to stop the spread of infection. We can infer that culling only the infected population can only be successful if we can eradicate the infected population. Currently, we cannot easily identify infected birds from the poultry population. Culling both susceptible and infected birds led to near eradication of the infected population, and, due to the low number of susceptible birds, further spread of H5N6  would not be possible. Thus, culling both susceptible and infected birds is necessary to eliminate the spread of infection in the poultry population.

\section{Conclusion}

The control strategies we considered include isolation and treatment of infected birds (isolation model), preventive vaccination of poultry (vaccination model), and modified culling of infected and susceptible birds that are at high risk of infection (culling model). In the model where isolation and treatment of infected birds is used as strategy, we extended previous models by considering that some birds that were released from confinement did not recover successfully. In using preventive vaccination, we also included the waning effect of the vaccine (in the model). For the model that depopulates the infected and susceptible birds that are high-risk to infection, we represented culling rate function with respect to half-saturated incidence.

Our results suggest that, when the basic reproduction number ($\mathcal{R}_A$, $\mathcal{R}_T$, $\mathcal{R}_V$, and $\mathcal{R}_C$) for each model is below unity, then the disease-free equilibrium is locally asymptotically stable. All four mathematical models presented here exhibited a forward bifurcation (Fig.~\ref{fig:bifurcation}), so lowering the basic reproduction number below 1 is sufficient to eliminate H5N6 from the poultry population.

In applying the optimal-control approach in the isolation strategy, we showed that isolation alone cannot prevent the spread of infection. Instead, it needs to be coupled with treatment so that the isolated birds can recover and heal from the infection. Figs.~\ref{fig:opt_vac_with_without}--\ref{fig:opt_vac_variedB} depict the importance of vaccine efficacy for a vaccination strategy to succeed in hindering the spread of H5N6. 

Depopulating the whole poultry population or mass culling during an outbreak is unacceptable for ethical, ecological and economic reasons \cite{butler_vaccination_2005}. However, various culling strategies have been considered by several studies, where they obtained that a threshold policy for culling can prevent overkilling of birds \cite{chong_avian-only_2016,gulbudak_forward_2013,gulbudak2014coexistence}. In this work, we examine the modified culling strategy, which includes depopulation of not only infected birds but also susceptible birds that are at high risk of infection. The depopulation of both susceptible and infected birds is an effective strategy to put an end to the spreading of avian influenza (as shown in Figs.~\ref{fig:opt_cul_with_without}-\ref{fig:opt_cul_u67}). Through application of optimal control to the culling strategy, we suggest culling at least three times during the first 10 days of the outbreak. 

Computing the basic reproduction number can contribute to decision-making in order to identify which parameters will help in inhibiting the transmission of H5N6.  By applying the optimal-control approach to different intervention strategies against H5N6, we have shown that culling of both infected and susceptible birds that are at high risk of infection is a better control strategy in prohibiting an outbreak and avoiding further recurrence of the infection from the population than confinement and vaccination. Every intervention strategy against H5N6 has advantages and disadvantages, but proper execution and appropriate application is a significant factor in achieving a desirable outcome.

\section*{Acknowledgment}
Lucido acknowledges the support of the Department of Science and Technology-Science Education Institute (DOST-SEI), Philippines for the ASTHRDP Scholarship grant. Lao holds research fellowship from De La Salle University. RS?\ is supported by an NSRC Discovery Grant. For citation purposes, please note that the question mark in ``Smith?" is part of his name.

\bibliographystyle{elsarticle-num}
\bibliography{sample}

\newpage

{\appendix

\section{Variables and parameters}
\label{tab:list}
Here, we describe each variable and parameter that we used in the AIV model, isolation model, vaccination model, and culling model.
\begin{table}[ht]
%\caption{List of variables and parameters in the models.} 
{\begin{tabular}{@{}cl@{}} \hline
\textbf{Notation} & \textbf{Description or Label} \\ 
\hline%\colrule
$S(t)$ & Susceptible birds  \\
$I(t)$ & Infected birds  \\
$T(t)$ & Isolated birds \\
$R(t)$ & Recovered birds  \\
$V(t)$ & Vaccinated birds \\
$N(t)$ & Total bird population  \\
$\Lambda$ & Constant birth rate of birds \\
$\mu$ & Natural death rate of birds \\
$\beta$ & Rate at which birds contract avian influenza \\
$H$ & Half-saturation constant for birds  \\
$\delta$ & Additional disease death rate due to avian influenza\\
$p$ & Prevalence rate of the vaccination program  \\
$\phi$ & Efficacy of the vaccine \\
$\omega$ & Waning rate of the vaccine \\
$\psi$ & Isolation rate of identified infected birds  \\
$\gamma$ & Releasing rate of birds from isolation \\
$f$ & Proportion of recovered birds from isolation  \\
$c_s$ & Culling frequency for susceptible birds \\
$c_i$ & Culling frequency for infected birds \\
$\tau_s(I)$ & Culling rate of susceptible birds \\
$\tau_i(I)$ & Culling rate of infected birds \\ \hline
%\botrule
\end{tabular}}
%\label{tab:list}
\end{table}

The initial conditions are based on Philippine Influenza A (H5N6) outbreak report given by the OIE \cite{noauthor_oie_2018} together with the assumed parameter values.
\begin{table}[ht]
%\label{tab:parameter_values}
%\tbl{Model parameter values}
{\begin{tabular}{@{}lclc@{}} \hline
 \textbf{Definition} & { {\textbf{Symbol}}}  &   { {\textbf{Value}}}  &  { {\textbf{Source}}} \\ \hline
{ {Constant birth rate of birds}} & $\Lambda$ &  $\frac{2 \, 060}{365}$ { {per day}} &  { {\cite{chong_avian-only_2016}}}\\
{ {Natural mortality rate}} & $\mu$ & ${ {3.4246 \times 10^{-4}}}$ { {per day}} & { {\cite{liu_nonlinear_2017}}} \\
{ {Transmissibility of the disease}} & $\beta$ & ${ {0.025}}$ { {per day}} & { {Assumed}} \\
{ {Half-saturation constant for birds}} & $H$ &  {180 000 birds} & { {\cite{lee_transmission_2018}}} \\
{ {Disease induced death rate of poultry}} & $\delta$ & ${ {4 \times 10^{-4}}}$ { {per day}}& { {\cite{liu_nonlinear_2017}}} \\
{ {Prevalence rate of vaccination program}} & $p$ & ${ {0.50}}$ & { {Assumed}} \\
{ {Vaccine efficacy}} &$\phi$ & ${ {0.90}}$ &  { {Assumed}} \\
{ {Waning rate of the vaccine}} &$\omega$ & ${ {0.00001}}$ per day&  { {Assumed}} \\
{ {Isolation rate of identified infected birds}} &  $\psi$ & ${ {0.01}}$ { {per day}}& { {Assumed}}  \\
{ {Releasing rate of birds from isolation}} & $\gamma$ & ${ {0.09}}$ { {per day}}& { {Assumed}} \\
{ {Proportion of fully-recovered birds from isolation}} & $f$ & ${ {0.50}}$  & { {Assumed}} \\
{ {Culling frequency for susceptible birds}} & $c_s$ & $ \frac{1}{60} $ { {per day}} & { {Assumed}} \\
{ {Culling frequency for infected birds}} & $c_i$ & $  \frac{1}{7}$ {per day}&{Assumed}\\ \hline %\botrule
\end{tabular}} 
\end{table}}

\section{Non-existence of backward bifurcation}
\subsection{Vaccination}
\label{sec:vac_back}

In showing that backward bifurcation does not exist for the vaccination model, we have $I^{***}= \dfrac{-b \pm \sqrt{b^2 -4ac}}{2a}$

where 

\begin{equation}
\label{eq:vac abc 1}
\begin{aligned}
     a & = - (\mu + \delta) [\mu \beta(1 - \phi) + (\mu + \omega) (\mu + \beta)+ \beta^2 (1-\phi)], \\
    b &  = \beta^2 \Lambda (1-\phi) + \mu H (\mu + \delta) (\mu + \omega) (\mathcal{R}_V -1) \\
    & \qquad - (\mu + \delta) H [\mu \beta (1-\phi) + (\mu + \beta) (\mu + \omega)], \\
    c &  = \mu H^2 (\mu + \delta) (\mu + \omega) (\mathcal{R}_V -1). 
\end{aligned}
\end{equation}

\begin{theorem}
\label{thm:vac endemic eq}
The vaccination model (\ref{eq:vac model}) has no endemic equilibrium when $\mathcal{R}_V \leq 1$, and has a unique endemic equilibrium when $\mathcal{R}_V > 1$.
\end{theorem}

\begin{proof}
We obtain two possible endemic equilibria $E_{V_1}^{*}$ and $E_{V_2}^{*}$for the vaccination model. From  (\ref{eq:vac abc 1}), we establish the relationship between $\mathcal{R}_V$ and $c$ such that

\begin{align*}
%\label{eq:vac RV and c}
\mathcal{R}_V > 1 \,\, \Leftrightarrow \,\, c > 0, \qquad \mathcal{R}_V = 1 \,\, \Leftrightarrow \, \, c = 0, \qquad \mathcal{R}_V < 1 \, \Leftrightarrow \,\, c < 0
\end{align*}

From (\ref{eq:vac abc 1}), it is clear that $a<0$. Now, we consider the following case when $c>0$, when $b>0$ and $c = 0$ or $(b^2 - 4ac) = 0$, and when $c<0$, $b>0$, and $(b^2 - 4ac) > 0$. \\

Case 1: $c>0$\\
When $c>0$, we have $\mathcal{R}_V>1$. Since $a<0$, it follows that

\begin{align*}
    & I_{1}^{***} = \dfrac{- b + \sqrt{b^2 - 4ac}}{2a} < 0
    & I_{2}^{***} = \dfrac{- b - \sqrt{b^2 - 4ac}}{2a} > 0
\end{align*}

and when $\mathcal{R}_V > 1$ the infected population ($I_{1}^{***}$) of the endemic equilibrium ($E_{V_1}^{*}$) does not exist and we have a unique endemic equilibrium $E_{V_2}^{*}$.\\

Case 2: $b>0$ and either $c = 0$ or $b^2 - 4ac=0$ \\
Given that $b>0$, we consider the case when $c = 0$ and when $b^2 - 4ac = 0$.\\

Case 2A: $c = 0$ \\
Since $c=0$ then $I_{1}^{***} = 0$ and $I_{2}^{***} > 0$. Note that $I_{1} = 0$ leads to the disease-free equilibrium. Hence, if $b>0$ and $c = 0$ then $I_{2}^{***} > 0$ and we have a unique endemic equilibrium $E_{V_2}^{*}$. \\

Case 2B: $b^2 - 4ac = 0$\\
Considering that $b^2 - 4ac = 0$, it follows that $I_{1}^{***} = I_{2}^{***}$ and $I_{1}^{***}, \, I_2^{***}>0$. Thus, if $b>0$ and $b^2 - 4ac = 0$, then we have a unique endemic equilibrium $E_{V_1}^{***} = E_{V_2}^{***}$. \\

Case 3: $c<0$, $b>0$, and $b^2 - 4ac >0$\\
From the assumption that $a<0$ and $c<0$, it follows that 
\begin{align*}
    & I_{1}^{***} = \dfrac{- b + \sqrt{b^2 - 4ac}}{2a} > 0
    & I_{2}^{***} = \dfrac{- b - \sqrt{b^2 - 4ac}}{2a} > 0.
\end{align*}
Thus, we have two endemic equilibria $I_{1}^{***}$ and $I_{2}^{***}$ which implies that backward bifurcation may possibly occur whenever $c<0$, $b>0$, and $b^2-4ac>0$.

However, given the values of $b$ and $c$, we can show that when $c<0$ we cannot obtain $b>0$ which we prove by contradiction. Suppose that $c<0$ and by definition of $p$ and $\phi$, the value of both parameters ranges from $0$ to $1$, that is $0 \leq p \leq 1$ and $0 \leq \phi \leq 1$. So, from (\ref{eq:vac abc 1}) it follows that $\Lambda \beta  < \dfrac{\mu H \Delta}{\Theta}$ where we define $\Theta = (\mu + \omega - p \mu \phi)$ and $\Delta = (\mu + \delta)(\mu + \omega)$. 

Using (\ref{eq:vac abc 1}) with $b>0$ we get $\Lambda \beta^2 (1-\phi) + \Lambda \beta \Theta  >  2 \mu H \Delta + \beta H \Delta +  \mu H \beta (\mu+\delta) (1-\phi)$. By simplifying, we obtain 

\begin{align}
\dfrac{\mu \beta (\mu + \omega) (1-\phi)}{\Theta}  & >  \mu (\mu + \omega) + \beta  (\mu + \omega) +  \mu  \beta  (1-\phi) \label{eq:vac bc2}.
\end{align}

As mentioned above $0 \leq \phi \leq 1$, so we consider the minimum and maximum value of $\phi$ into the inequality in (\ref{eq:vac bc2}).\\
Case 3A: Let $\phi = 0$.\\
Assuming that $\phi = 0$ so that $\Theta = (\mu + \omega_b)$ and we simplify (\ref{eq:vac bc2}) as follows:

\begin{align*}
     0  & >  \mu (\mu + \omega) + \beta  (\mu + \omega). %\label{eq:vac I contradiction3}   
\end{align*}

Since all the parameter $\mu$, $\omega$, $\beta \geq 0$, it implies that $ 0 \leq \mu (\mu + \omega) + \beta  (\mu + \omega)$. Thus, we have a contradiction. Hence, for $\phi=0$ and when $c<0$ it follows that $b \ngtr0$.\\

Case 3B: Let $\phi =1$.\\
When $\phi=1$, we can simplify (\ref{eq:vac bc2}) into

\begin{equation*}
%\label{eq:vac II contradiction}
    0 >  \mu H \Delta + \beta H \Delta. 
\end{equation*}

Similarly, given that the parameter $\mu$, $H$, $\omega$, $\delta$, and $\beta \geq 0$, it signifies that we have a contradiction. Thus, when $\phi =1$ and $c<0$ then $b \ngtr 0$.\\

From Case 3A and Case 3B, we have shown that for all values of $\phi$ as it ranges from $0$ to $1$, $b \ngtr 0$ whenever $c<0$. Results above suggest that two endemic equilibria does not exist when $\mathcal{R}_V < 1$, since the condition $c<0$, $b>0$, and $b^2 - 4ac>0$, cannot be satisfied. From Cases 1 to 3, it is evident that the vaccination model has no endemic equilibrium when $\mathcal{R}_V < 1$ and a unique endemic equilibrium when $\mathcal{R}_V \geq 1$.
\end{proof}

\subsection{Culling}
To show that the backward bifurcation does not exist we solve for $I^{****}= \dfrac{-b \pm \sqrt{b^2 -4ac}}{2a}$ such that

\begin{equation}
\label{eq:cul abc1}
\begin{aligned}
    & a = - (\mu + \delta + c_i) (\mu + c_s + \beta), \\
    & b = \mu H (\mu + \delta)(\mathcal{R}_C - 1) - c_i \mu H - H (\mu + \delta) (\mu + c_s + \beta), \\
    & c = \mu H^2 (\mu + \delta)(\mathcal{R}_C - 1). 
\end{aligned}
\end{equation}

\begin{theorem}
\label{thm:cul endemic eq}
The culling model (\ref{eq:cul model}) has no endemic equilibrium when ${\mathcal{R}_C < 1}$, and has a unique endemic equilibrium when $\mathcal{R}_C > 1 $.
\end{theorem}

\begin{proof}
We begin with applying the quadratic formula to obtain

\begin{equation*}
%\label{eq:cul 2i}
\begin{aligned}
    & I_{1}^{****} = \dfrac{- b + \sqrt{b^2 - 4ac}}{2a}, \qquad
    & I_{2}^{****} = \dfrac{- b - \sqrt{b^2 - 4ac}}{2a}.
\end{aligned}    
\end{equation*}

From (\ref{eq:cul abc1}),  $a < 0$ and we consider cases where $\mathcal{R}_C < 1$, $\mathcal{R}_C = 1$, and $\mathcal{R}_C > 1$.\\

Case 1: $\mathcal{R}_C<1$ \\
When $\mathcal{R}_C$ is below unity, it follows that $c<0$ and $b<0$.\\

Given that $a<0$ and $c<0$, we can say that $4ac > 0$ and we get the following: 

\begin{align*}
    & I_{1}^{****} = \dfrac{- b + \sqrt{b^2 - 4ac}}{2a} < 0,
    & I_{2}^{****} = \dfrac{- b - \sqrt{b^2 - 4ac}}{2a} < 0.
\end{align*}

Thus, in our case when $\mathcal{R}_C < 1$, we have no endemic equilibrium.\\

Case 2: $\mathcal{R}_C = 1$\\
When $\mathcal{R}_C =1$, it results to $c = 0$ and $b < 0$. Assuming that $c = 0$, then we obtain $4ac = 0$ and it follows that $\sqrt{b^2 -4ac} = b$. Since $a< 0$, we realize that

\begin{align*}
    & I_{1}^{****} = \dfrac{- b + b}{2a} = 0,
    & I_{2}^{****} = \dfrac{- b - b}{2a} < 0.
\end{align*}

Hence, when $\mathcal{R}_C = 1$, we have no endemic equilibrium.\\

Case 3: $\mathcal{R}_C > 1$\\
When $\mathcal{R}_C$ is above the unity, it follows that $c > 0$. Given that $a< 0$ and $c>0$, then we get

\begin{align*}
    & I_{1}^{****} = \dfrac{- b + \sqrt{b^2 - 4ac}}{2a} < 0
    & I_{2}^{****} = \dfrac{- b - \sqrt{b^2 - 4ac}}{2a} > 0.
\end{align*}

Hence, when $\mathcal{R}_C > 1$ we have $I_{2}^{****} > 0$ and a unique endemic equilibrium $E_{C_2}^{****}$.
\end{proof}

\end{document}